\documentclass[11pt]{article}

\usepackage{geometry}                			
\usepackage{amssymb}
\usepackage{amsmath}
\usepackage{palatino}
\usepackage{graphicx}
\usepackage[english]{babel}
\usepackage{amscd}
\usepackage{amsthm}
\usepackage{color}
\usepackage{enumerate}
\usepackage{hyperref}
\usepackage{pgfplots}
\usepackage{tikz}
\usepackage{bm}

\usepackage{subfigure}
\usepackage{caption}

\usetikzlibrary{plotmarks}
\pgfplotsset{compat=1.16}

\newtheorem{theorem}{Theorem}
\newtheorem{conjecture}{Conjecture}

\newtheorem{lemma}{Lemma}

\newtheorem{definition}{Definition}
\newtheorem{assumption}{Assumption}

\newtheorem{remark}{Remark}

\newcommand{\xb}{\mathbf{x}}
\newcommand{\xt}{\xb_T}
\newcommand{\bb}{\mathbf{b}}
\newcommand{\yb}{\mathbf{y}}
\newcommand{\yt}{\yb_T}
\newcommand{\zb}{\mathbf{z}}
\newcommand{\wb}{\mathbf{w}}
\newcommand{\ji}{\xb_T^{j|i}}

\newcommand{\xtj}{\xb_T^{(j)}}
\newcommand{\ytj}{\yb_T^{(j)}}
\newcommand{\mR}{\mathcal{R}}
\newcommand{\mV}{\mathcal{V}}
\newcommand{\mF}{\mathcal{F}}
\newcommand{\mE}{\mathcal{E}}
\newcommand{\mG}{\mathcal{G}}
\newcommand{\mH}{\mathcal{H}}
\newcommand{\mC}{\mathcal{C}}
\newcommand{\mS}{\mathcal{S}}
\newcommand{\mN}{\mathcal{N}}
\newcommand{\mT}{\mathcal{T}}
\newcommand{\mL}{\mathcal{L}}
\newcommand{\eps}{\varepsilon}

\title{A Fragile multi-CPR Game\thanks{Research was supported by the Hellenic Foundation for Research and Innovation (H.F.R.I.) under the ?First Call for H.F.R.I. Research Projects to support Faculty members and Researchers and the procurement of high-cost research equipment grant? (Project Number: HFRI-FM17-2436).}}

\author{
Christos~Pelekis\thanks{School of Electrical and Computer Engineering, National Technical University of Athens, Zografou, Greece, 15780, e-mail: pelekis.chr@gmail.com} \and   Panagiotis~Promponas\thanks{School of Electrical and Computer Engineering, National Technical University of Athens, Zografou, Greece, 15780, e-mail: ppromponas@mail.ntua.gr} \and 
Juan Alvarado\thanks{KU Leuven, Department of Computer Sciences,  Celestijnenlaan 200A, 3001, Belgium, e-mail: Juan.Alvarado@cs.kuleuven.be}
\and    Eirini~Eleni~Tsiropoulou\thanks{Department
of Electrical and Computer Engineering, University of New Mexico, New Mexico,
USA, 87131, e-mail: eirini@unm.edu} \and  Symeon~Papavassiliou\thanks{School of Electrical and Computer Engineering, National Technical University of Athens, Zografou, Greece, 15780, e-mail: papavass@mail.ntua.gr}}

\begin{document}

\maketitle

\begin{abstract}
A Fragile CPR Game is an instance of a resource sharing game where a common-pool resource, which is prone to failure due to overuse, is shared among several players. Each player has a fixed initial endowment and is faced with the task of investing in the common-pool resource without forcing it to fail. The return from the common-pool resource is subject to uncertainty and is perceived by the players in a prospect-theoretic manner. It is shown in Hota et al.~\cite{Hota} that, under some mild assumptions, a Fragile CPR Game admits a unique Nash equilibrium.  In this article we investigate an extended version of a Fragile CPR Game, in which  players are allowed to share multiple common-pool resources that are also prone to failure due to overuse. We refer to this game as a Fragile multi-CPR Game. Our main result states that, under some mild assumptions, 
a Fragile multi-CPR Game admits a Generalized Nash equilibrium. Moreover, we show that, when there are more players than common-pool resources, 
the set consisting of all Generalized Nash 
equilibria of 
a Fragile multi-CPR Game is of Lebesgue measure zero. 
\end{abstract}



\section{Prologue, related work and main results}

In this article we shall be concerned with a    
\emph{resource sharing game}. Such  games model 
instances in which a common-pool resource (henceforth CPR), which is prone to failure due to overuse, is shared among several users who are addressing the problem of choosing how much to exploit from / invest in  the CPR  without forcing it to fail. Resource sharing games arise in a variety of problems ranging from economics to computer science. 
Examples of CPRs include arable lands, forests, fisheries, groundwater basins, spectrum and computing resources, the atmosphere, among many others. Such CPRs are, on the one hand, usually regenerative but, on the other hand, subject to 
failure when several agents exploit the resource in  an unsustainable manner.  Each agent exploits / invests in the CPR  
in order to obtain an individual benefit. However, 
it has been observed that actions which are 
individually rational  (e.g. Nash equilibria) may result in outcomes that are collectively irrational, thus giving rise to 
a particular social dilemma known as ``the tragedy of the commons" (see~\cite{Hardin}). 
It is thus of interest to 
investigate equilibrium points of resource sharing games, 
in order to better 
understand situations where such a social dilemma arises. 
This is a topic that has drawn considerable attention, both from a theoretical and a practical perspective.  
We refer the reader to~\cite{Aflaki, Budescu_et_al, Hota, Hota_Sundaram_2, Keser_Gardner, Ostrom, Ostrom_et_al, Vamvakas, VamvakasTP, Walker_Gardner} for   applications, variations, and for further references on resource sharing games. 
Let us remark that most results in the literature 
appear to focus on games 
in which players invest on a single CPR. 
In this article we investigate a  resource sharing 
game in which players are allowed to invest in 
several CPRs, whose performances are 
mutually independent. To the best of our knowledge our work appears to be 
among the first to consider resource sharing games on more than one CPR.

We shall be interested in a multi-version  of a particular resource sharing game, which is referred to as a  \emph{Fragile CPR Game}. It is initially  introduced in~\cite{Hota} and  is played by several players, each of whom has a fixed initial endowment and must decide how much to invest in the CPR without forcing it to fail. The return from the CPR is subject to uncertainty, and  is perceived by the players in a prospect-theoretic manner.  
It is shown in~\cite{Hota} that a Fragile CPR Game admits a unique Nash equilibrium.  In this article we 
focus on  an extended version of a Fragile CPR Game in which players are allowed to share multiple CPRs. 
We refer to the corresponding game as a  
\emph{Fragile multi-CPR Game} and investigate its 
\emph{Generalized Nash equilibria}. 
Our main result states that the set consisting of all Generalized Nash equilibria of a  Fragile multi-CPR Game is non-empty and, when there are more players than CPRs, ``small" in a measure-theoretic sense.  
In the next subsection we introduce  the 
Fragile CPR Game and state the main result from~\cite{Hota}. We then proceed, in Subsection~\ref{def:multiCPR}, with defining the Fragile multi-CPR Game, which is the main target of this work, and stating our main results.

\subsection{Fragile CPR game}

Throughout the text, given a positive integer $n$, we denote by $[n]$ the set  $\{1,\ldots, n\}$.
In this article we extend a particular resource sharing game  to the case where the players are allowed to share multiple resources, by determining how to distribute/invest their initial fixed endowment in the available CPRs. The resource sharing game under consideration is referred to as a   \emph{Fragile CPR Game},  
and may be seen as a prospect-theoretic version of the \emph{Standard CPR Game}, introduced in~\cite[p.~109]{Ostrom_et_al}. 

The Fragile CPR Game is introduced in~\cite{Hota}, and is played by $n$ players, who are assumed to be indexed by the set $[n]$.  It is also assumed that there is a single CPR, and  each player has to decide  how much to invest in the CPR. Each player has an available endowment, which, without loss of generality, is assumed to be equal to $1$.  
Every player, say $i\in [n]$, invests  an amount $x_i\in [0,1]$ in the CPR. The total investment of all players in the CPR is denoted $\xt = \sum_{i\in [n]} x_i$. 
The return from the CPR is  subject to uncertainty, that is there is a probability $p(\xt)$ that the CPR will 
fail, and this probability depends on the total investment of the players in the CPR.  In case the CPR fails, 
the players lose their investment in the CPR. In case the CPR does not fail, then there is a \emph{rate of return} from the CPR which depends on the total investment of all players, and is denoted by $\mR(\xt)$. The rate of return is assumed to satisfy $\mR(\xt)>1$, for all $\xt$. 

In other words, player $i\in [n]$ gains $x_i\cdot \mR(\xt)-x_i$  with probability $1-p(\xt)$, and gains $-x_i$ with probability $p(\xt)$. 
The situation is modelled through a prospect-theoretic perspective, in the spirit of~\cite{Kahneman_Tversky}. More precisely, let 
$x^{(i)} = \sum_{j \in [n]\setminus \{i\}} x_j$; hence it holds  $x_i + x^{(i)} = \xt$. Then the utility of player $i\in [n]$ is given by the following utility function:

\begin{equation}\label{utility_Hota}
\mV_i(x_i, x^{(i)})=
\begin{cases}
(x_i\cdot (\mR(\xt)-1))^{a_i}, &\text{ with probability } 1-p(\xt),\\
-k_i x_i^{a_i}, &\text{ with probability } p(\xt).
\end{cases}
\end{equation}

The parameters $k_i$ and $a_i$ are fixed and player-specific.  
Let us  note that the  parameter $k_i$ may be thought of as capturing the ``behaviour" of each player. 
More precisely, when $k_i >1$ then a player weighs losses more than gains, a behaviour which is referred to  as 
``loss averse". On the other hand, when $k_i\in [0,1]$
then a player weighs gains more than losses, a behaviour which is referred to as ``gain seeking".  Capturing  behaviours of this type among players constitutes a central aspect of prospect theory (see, for example,~\cite{Wakker}). Notice that when $k_i=1$ and $a_i=1$ then player $i\in [n]$ is \emph{risk neutral}.

Each player of the Fragile CPR game is an expected utility maximizer, and therefore chooses $x_i\in [0,1]$ that maximizes the expectation 
of $\mV(x_i,x^{(i)})$, i.e, that maximizes the \emph{utility} of player $i\in [n]$ which is given by 
\[
\mathbb{E}\left(\mV_i(x_i, x^{(i)})\right) = x_i^{a_i} \cdot \mF_i(\xt) \, ,
\]
where 
\begin{equation}\label{eff_rate}
\mF_i(\xt) =  (\mR(\xt)-1)^{a_i} \cdot (1-p(\xt)) - k_i\cdot p(\xt) 
\end{equation}
is the~\emph{effective rate of return} to payer $i\in [n]$. 

The main result in~\cite{Hota} establishes, among other things, the existence of a unique Nash equilibrium for the Fragile CPR game, 
provided the following hold true. 

\begin{assumption}\label{first_ass}
Consider a Fragile CPR game that satisfies the following properties. 
\begin{enumerate}
\item It holds $p(0)=0$ and $p(\xt)=1$, whenever $\xt\ge 1$.  
\item $a_i \in (0,1]$ and $k_i >0$, for all $i\in [n]$. 
\item For all $i\in [n]$ and all $\xt\in (0,1)$ it holds $\frac{\partial}{\partial \xt} \mF_{i}(\xt), \frac{\partial^2}{\partial \xt^2} \mF_{i}(\xt) <0$, where $\mF_i$ is given by~\eqref{eff_rate}.  
\end{enumerate}
\end{assumption}

In other words, the third condition in 
Assumption~\ref{first_ass} states that the effective rate 
of return of all players is a strictly decreasing and concave function.
An example of an effective rate of return $\mF_{i}$ satisfying the conditions of Assumption~\ref{first_ass} is obtained 
by choosing $a_i < 1/2$,  $p(\xt) =\xt^2$, and $\mR(\xt)=2-e^{\xt-1}$, as can be easily verified. 

Before proceeding with the main result from~\cite{Hota}, let us recall here the notion of Nash equilibrium, adjusted to the setting of the Fragile CPR Game. 

\begin{definition}(Nash Equilibrium)
\label{def:NE}
A \emph{Nash equilibrium} for a Fragile CPR Game is a 
strategy profile $(x_1,\ldots,x_n)\in [0,1]^n$ such that for all $i\in [n]$ it holds:
\[
\mathbb{E}\left(\mV(x_i,x^{(i)}) \right) \ge \mathbb{E}\left(\mV(z_i,x^{(i)}) \right) \, , \text{ for all } \, z_i \in [0,1] \, .
\]
\end{definition}

In other words, $(x_1,\ldots,x_n)\in [0,1]^n$ is a Nash equilibrium for a Fragile CPR Game if no player can increase her utility  by unilaterally changing 
strategy. The main result in Hota et al.~\cite{Hota} reads as follows. 

\begin{theorem}[\cite{Hota}]
\label{Hota}
Consider a Fragile CPR Game that satisfies Assumption~\ref{first_ass}. Then the game admits a \emph{unique} Nash equilibrium. 
\end{theorem}

We now proceed with defining the \emph{Fragile multi--CPR Game}, whose equilibria are the main target of the present article.  

\subsection{Fragile multi-CPR game}\label{def:multiCPR}

In this article we introduce and investigate a multi-version of the Fragile CPR game.  
In order to be more precise, we need some extra piece of notation. 
If $m$ is a positive integer, let $C_m$ denote the set:
\begin{equation}\label{strategy_space}
C_m = \left\{ (x_1,\ldots, x_m)\in [0,1]^m : \sum_{i\in [m]} x_i \le 1 \right\} \, . 
\end{equation}
Moreover, let $\mC_n$ denote the Cartesian product $\prod_{i\in [n]} C_m$ and let $\mC_{-i} = \prod_{[n]\setminus \{i\}}C_m$ denote the Cartesian product 
obtained from $\mC_n$ by deleting its $i$-th component.  Elements in $\mC_{-i}$ are denoted by $\xb_{-i}$, as is customary, and  
an element $\xb=(\xb_1,\ldots,\xb_n)\in\mC_n$ is 
occasionally written $\xb=(\xb_i,\xb_{-i})$, for $i\in [n]$, $\xb_i\in C_m$ and $\xb_{-i}\in\mC_{-i}$. We now proceed with defining the \emph{Fragile multi-CPR 
Game}.

Suppose that there are $n$ players, indexed by the set $[n]$, each having an \emph{initial endowment} equal to $1$. Assume further 
that there are $m$ available CPRs, where $m\ge 1$ is an integer. Every player has to decide how much to invest in each CPR. More precisely, 
every player, say $i\in [n]$, chooses an element $\xb_i = (x_{i1},\ldots, x_{im}) \in C_m$ and  invests $x_{ij}$
in the $j$-th CPR. 
Given strategies $\xb_i=(x_{i1},\ldots,x_{im})\in C_m, i\in [n]$, of the players and an integer $j\in [m]$, set   
\begin{equation}\label{sum_xs}
\xtj = \sum_{i\in [n]} x_{ij} \quad \text{ and } \quad \ji = \sum_{\ell\in [n]\setminus \{i\}}x_{\ell j}  \, . 
\end{equation}
Hence it holds $\xtj = x_{ij} + \ji$, for all $i\in [n]$. 
In other words, $\xtj$ equals the total investment of 
the players in the $j$-th CPR and $\ji$ equals the total investment of all players except player $i$ in the $j$-th CPR.  
As in the case of the Fragile CPR Game, we assume that the performance of each CPR is subject to uncertainty, and that each CPR has a corresponding rate of return, both depending on the total investment of the players in each CPR. More precisely, 
for $j\in [m]$, let $\mR_j(\xtj)$ denote the \emph{return rate} of the $j$-th CPR and let $p_j(\xtj)$ denote the \emph{probability that the $j$-th CPR fails}. We assume that $\mR_j(\xtj)>1$ holds true, for all $\xtj$. 

The \emph{utility} of player $i\in [n]$ from the $j$-th CPR is given, as in the case of the Fragile CPR game, via the following prospect-theoretic utility function:

\begin{equation}\label{utility_2}
\mV_{ij}(x_{ij}, \ji)=
\begin{cases}
(x_{ij}\cdot (\mR_j(\xtj)-1))^{a_i}, &\text{ with probability } 1-p_j(\xtj),\\
-k_i x_{ij}^{a_i}, &\text{ with probability } p_j(\xtj).
\end{cases}
\end{equation}

We assume that the performance of each CPR is \emph{independent} of the performances of all  remaining CPRs. 
Players in the Fragile multi-CPR Game are expected utility maximizers. 
If player $i\in [n]$ plays the vector $\xb_i = (x_{i1},\ldots, x_{im})\in C_m$, and the rest of the players play $\xb_{-i}\in\mC_{-i}$ then her 
expected utility from the $j$-th CPR is equal to 
\begin{equation}\label{1_dimensional}
\mE_{ij}(x_{ij};\ji) :=\mathbb{E}\left(\mV_{ij}(x_{ij}, \ji)\right) =  x_{ij}^{a_i}\cdot \mF_{ij}(\xtj) \, ,
\end{equation}
where 
\begin{equation}\label{effective_rr}
\mF_{ij}(\xtj):= (\mR_{j}(\xtj)-1)^{a_i} (1-p_j(\xtj)) - k_i  p_j(\xtj)
\end{equation}
is the \emph{effective rate of return} to the $i$-th player from the $j$-th CPR.
Notice that, since we assume that the performance of each CPR is independent of the performances of the remaining CPRs, $\mE_{ij}$ depends only on the values of $x_{ij},\ji$ and does not depend on the values of $x_{ik},\mathbf{x}_{T}^{k|i}$, for $k\neq j$. In other words, 
the (total) prospect-theoretic \emph{utility} of player $i\in [n]$ in the Fragile multi-CPR Game is given by: 
\begin{equation}\label{utility}
\mV_i(\xb_i; \xb_{-i}) = \sum_{j\in [m]} \mE_{ij}(x_{ij};\ji) \, .
\end{equation}

In this article we establish the existence of a Generalized Nash equilibrium for the Fragile multi-CPR game, provided the following holds true.  

\begin{assumption}\label{ass_prob}
Consider a Fragile multi-CPR Game that satisfies the following properties:
\begin{enumerate}
\item For every $j\in [m]$ it holds $p_j(0)=0$ and  $p_j(\xtj) = 1$, whenever 
$\xtj\ge 1$. 
\item It holds $a_i \in (0,1]$ and $k_i>0$, for all $i\in [n]$.
\item For all $i\in [n]$ and all $j\in [m]$ it holds $\frac{\partial}{\partial \xtj} \mF_{ij}(\xtj), \frac{\partial^2}{\partial (\xtj)^2} \mF_{ij}(\xtj) <0$, where $\mF_{ij}$ is given by~\eqref{effective_rr}.
\end{enumerate} 
\end{assumption}

Notice that, similarly to the Fragile CPR Game, the third condition in Assumption~\ref{ass_prob} states that  the effective rate of return of every player from any CPR is a \emph{strictly decreasing and concave} function. An example of an effective rate of return satisfying Assumption~\ref{ass_prob} is obtained by 
choosing, 
for $j\in [m]$, the return rate of the $j$-th CPR to be equal to $\mR_j(\xtj) = c_j +1$, where $c_j>0$ is a constant, and 
the  probability that the $j$-th CPR fails to be a strictly increasing and convex, on the interval $[0,1]$, function such that $p_j(\xtj)=1$, when $\xtj\ge 1$. 

Before stating our main result, let us proceed with 
recalling the notion of Generalized Nash equilibrium (see~\cite{Facchinei_Kanzow}). 

Consider the, above-mentioned, Fragile multi-CPR Game, 
denoted $G$.  
Assume further that, for each player 
$i\in [n]$, there exists a correspondence $\vartheta_i: \mC_{-i}\to 2^{C_m}$ mapping every element $\xb_{-i}\in\mC_{-i}$ to a set $\vartheta_i(\xb_{-i})\subset C_m$. The set-valued correspondence $\vartheta_i$ is referred to as a \emph{constraint policy} and may be thought of as determining the set of strategies that are feasible for player $i\in [n]$, given  $\xb_{-i}\in\mC_{-i}$.
We refer to the tuple $(G,\{\vartheta_i\}_{i\in [n]})$
as the \emph{Constrained Fragile multi-CPR Game} with 
constraint policies $\{\vartheta_i\}_{i\in [n]}$. Corresponding to a constrained game is the following notion of \emph{Constrained Nash equilibrium} (or \emph{Generalized Nash equilibrium}):

\begin{definition}[GNE]\label{def:gne}
A \emph{Generalized Nash equilibrium} for a Constrained Fragile multi-CPR Game $(G,\{\vartheta_i\}_{i\in [n]})$ is a strategy profile $\xb^{\ast}=(\xb_1^{\ast},\ldots,\xb_n^{\ast})\in \mC_n$  such that
\begin{enumerate} 
 \item For all $i\in [n]$, it holds $\xb_i^{\ast} \in \vartheta_i(\xb_{-i}^{\ast})$,  for all $i\in [n]$, and 
\item For all $i\in [n]$, it holds $\mV_i(\xb_i^{\ast}; \xb_{-i}^{\ast}) \ge 
 \mV_i(\xb_i; \xb_{-i}^{\ast})$,  for all  $\xb_i\in \vartheta_i(\xb_{-i}^{\ast})$,
where $\mV_i(\,\cdot\, ;\,\cdot\,)$ is the utility function of the $i$-th player in a Fragile multi-CPR Game, given in~\eqref{utility}. 
\end{enumerate} 
\end{definition}

In other words,  $\xb^{\ast}=(\xb_1^{\ast},\ldots,\xb_n^{\ast})\in\mC_n$ is a GNE if no player can increase her 
utility by unilaterally changing her strategy to any other element of the set $\vartheta_i(\xb_{-i}^{\ast})$.
We may now proceed with stating our main results.

\begin{theorem}\label{unique_NE} 
Consider a Fragile multi-CPR game, $G$, with $n\ge 1$ players and $m\ge 1$ CPRs, which  satisfies 
Assumption~\ref{ass_prob}. Then there exist constraint 
policies $\{\vartheta_i\}_{i\in [n]}$ such that the 
Constraint Fragile multi-CPR Game $(G,\{\vartheta_i\}_{i\in [n]})$  admits a Generalized Nash equilibrium.  
\end{theorem}

Given Theorem~\ref{unique_NE}, it is natural to ask 
about the ``size" of the set consisting of all  
GNEs of a Fragile multi-CPR Game. 
Let us note that it is a well known fact that  Generalized Nash equilibrium problems 
tend to possess infinitely many GNEs (see~\cite[p.~192]{Facchinei_Kanzow}).  In the case of a single CPR, i.e., when $m=1$, the corresponding Constrained Fragile CPR Game admits a unique GNE. 

\begin{theorem}\label{one_CPR}
Consider a Fragile multi-CPR Game with $n\ge 1$ players and $m=1$ CPR satisfying Assumption~\ref{ass_prob}. 
Then the game admits a unique GNE. 
\end{theorem}

The proof of Theorem~\ref{one_CPR} is based upon 
a ``first order condition" which is satisfied 
by the best response correspondence in a Fragile 
multi-CPR Game. 
It turns out that the aforementioned 
``first order condition" gives rise to \emph{two types} 
of best responses for the players (see Theorem~\ref{kkt_thm} below).  In fact, we show that 
Theorem~\ref{one_CPR} is a consequence of 
a more general statement (i.e., Theorem~\ref{gen_one_CPR} below) which provides an upper bound on the numbers of GNEs in a Fragile multi-CPR Game, subject to the assumption that  
best response of every player is 
of the first type. 

For general $m$ we are unable to determine 
the exact ``size" of the set of GNEs. We conjecture its size is always finite. 
Our main result, which is valid when there are more 
players than CPRs, states that the set of GNEs is small in a measure-theoretic sense.  

\begin{theorem}\label{nr_of_GNE}
Consider a Fragile multi-CPR game, $G^{(2)}$, with $n\ge 1$ players and $m\ge 1$ CPRs, which  satisfies 
Assumption~\ref{ass_prob}. 
Assume further that $m\le n$, and let $\mN(G^{(2)})$ be the set 
consisting of all Generalized Nash equilibria of $G^{(2)}$. 
Then the $(n\cdot m)$-dimensional Lebesgue measure of $\mN(G^{(2)})$ is equal to zero. 
\end{theorem}

As mentioned already, and despite the fact that GNE problems tend to possess infinitely many solutions, we speculate that the  ``size" of the set $\mN(G^{(2)})$ in Theorem~\ref{nr_of_GNE} can be  reduced significantly. 

\begin{conjecture}\label{conj:1}
The set $\mN(G^{(2)})$ is finite. 
\end{conjecture}

\subsection{Brief outline of the proofs of main results}

The proofs of our main results are inspired from  the proof of Theorem~\ref{Hota}, given in~\cite{Hota}. Having said that, it should also be mentioned that in  
a Fragile multi-CPR Game certain additional technicalities arise that are substantially  different from those addressed in the 
proof of Theorem~\ref{Hota} in~\cite{Hota}. First and foremost, 
in a Fragile multi-CPR Game the strategy space of each player consists of  $m$-dimensional vectors, a setting which requires concepts and ideas from multi-variable calculus. 

In~\cite{Hota} the existence of a Nash equilibrium in a Fragile CPR Game is established in two ways: the first approach employs  Brouwer's fixed point theorem, and the second approach  
employs ideas from a particular class of games known as \emph{Weak Strategic Substitute Games} (see~\cite{Dubey}).  The first approach requires, among other things, the best response correspondence to be single-valued. 
The second approach requires the best-response correspondence to be decreasing. Both requirements may 
fail to hold true in a Fragile multi-CPR Game. Instead, 
we establish the existence of a Generalized Nash equilibrium for the Fragile multi-CPR Game by showing that it belongs to a particular class of ``convex constrained  games" which are known to possess Generalized Nash equilibria.

In~\cite{Hota} the uniqueness of the Nash equilibrium for a Fragile CPR Game is established by showing that a particular auxiliary function, corresponding to the fact that the best response correspondence satisfies a particular ``first order condition" (see~\cite[Eq.~(6), p.~142]{Hota} for the precise formulation of the condition), is decreasing. Similar auxiliary functions are employed in the proofs of Theorems~\ref{one_CPR} and~\ref{nr_of_GNE}.  However, the corresponding ``first order conditions" are more delicate to characterise, and we do so by 
employing the KKT conditions to the optimization program corresponding to the best response correspondence (i.e., Problem~\eqref{opt_problem2} below). This allows to describe the best responses via a system of equations, having unique solution, and results in two types of ``first order conditions" (see Theorem~\ref{kkt_thm} below). Having established the first order conditions in a Fragile multi-CPR Game, we complete the proofs of our main results by employing monotonicity properties of certain auxiliary functions, 
in a way which may be seen as an extension of the approach taken in the proof of Theorem~\ref{Hota} in~\cite{Hota}.

\subsection{Organization} 

The remaining part of our article is organised as follows. In Section~\ref{sec:2} we show that the utility function of each player in a Fragile multi-CPR Game is concave on a particular subset of the strategy space. 
In Section~\ref{sec:3} we prove Theorem~\ref{unique_NE}, namely, we show that a Fragile multi-CPR Game  admits a GNE. In Section~\ref{sec:4} we show that the best response of each player in a Fragile multi-CPR Game satisfies certain ``first order conditions", which are then used, in Section~\ref{sec:5}, in order to define suitable auxiliary functions whose monotonicity 
properties play a key role 
in the proofs of Theorems~\ref{one_CPR} and~\ref{nr_of_GNE}.  Theorem~\ref{one_CPR} 
is proven in Section~\ref{sec:6} and Theorem~\ref{nr_of_GNE} is proven in Section~\ref{sec:7}. 
In Section~\ref{sec:7.5} we show that a ``restricted" 
version of a Fragile multi-CPR Game admits finitely 
many GNEs, a result which is then employed in 
order to formulate a conjecture which is 
equivalent to Conjecture~\ref{conj:1}. 
Our paper ends with  
Section~\ref{sec:8} which includes some  concluding remarks and conjectures.

\section{Concavity of utility function}\label{sec:2}

In this section we show that the utility function, given by~\eqref{utility}, of each player in a Fragile multi-CPR Games is concave in some 
particular subset of $C_m$. Before proceeding 
with the details let us mention that this particular 
subset will be used to define the constraint policies 
in the corresponding Constrained Fragile multi-CPR Game.

We begin with the 
following result, which readily follows from~\cite[Lemma~1]{Hota}.   
Recall the definition of $\xtj$ and $\ji$, given in~\eqref{sum_xs}, and the 
definition of the effective rate of return, $\mF_{ij}$, given in~\eqref{effective_rr}. 

\begin{lemma}[see~\cite{Hota}, Lemma~1]
\label{VTP_lemma}
Let $i\in [n]$ and $\xb_{-i}\in\mC_{-i}$ be fixed. Then, for every $j\in [m]$, there exists a real number  
$\omega_{ij}\in (0,1)$ such that  
$\mF_{ij}(\xtj )>0$, whenever 
$\xtj \in (0,\omega_{ij})$, and $\mF_{ij}(\omega_{ij})=0$. 
Furthermore, provided that  $\ji< \omega_{ij}$, 
the function $\mE_{ij}(\,\cdot\, ; \ji)$ is concave in the interval  $(0,\omega_{ij}-\ji)$.
\end{lemma}
\begin{proof} 
We repeat the proof for the sake of completeness. 
Notice that $\mF_{ij}(0)>0$. 
Moreover, Assumption~\ref{ass_prob} implies that $\mF_{ij}(1)<0$. 
Since $\mF_{ij}$ is continuous, the intermediate value theorem implies that there exists $\omega_{ij}\in (0,1)$ such that $\mF_{ij}(\omega_{ij})=0$. Since $\mF_{ij}$ is assumed to be decreasing, the first statement follows, and we proceed with the proof of the second statement. 
To this end, notice that~\eqref{1_dimensional} yields  
\[
\frac{\partial^2}{\partial x_{ij}^2}
\mE_{ij}(x_{ij} ; \ji)  = a_i (a_i-1) x_{ij}^{a_i-2}  \mF_{ij}(\xtj) + 2a_ix_{ij}^{a_i-1}  \frac{\partial}{\partial x_{ij}} \mF_{ij}(\xtj) + x_{ij}^{a_i} \frac{\partial^2}{\partial x_{ij}^2} \mF_{ij}(\xtj) \, .
\]
Notice also that $\frac{\partial}{\partial x_{ij}} \mF_{ij}(\xtj) = \frac{\partial}{\partial \xtj} \mF_{ij}(\xtj)$ as well as  $\frac{\partial^2}{\partial x_{ij}^2} \mF_{ij}(\xtj) =\frac{\partial^2}{\partial( \xtj)^2} \mF_{ij}(\xtj)$. 
Moreover,  Assumption~\ref{ass_prob} guarantees that $\frac{\partial^2}{\partial x_{ij}^2} \mF_{ij}(\xtj), \frac{\partial}{\partial x_{ij}} \mF_{ij}(\xtj) <0$ as well as that $a_i-1\le 0$. Since $\mF_{ij}(\xtj)>0$ when 
$\xtj\in (0,\omega_{ij})$, we conclude that 
$\frac{\partial^2}{\partial x_{ij}^2}
\mE_{ij}(x_{ij} ; \ji) <0$ and therefore  
$\mE_{ij}(\,\cdot\, ;\ji)$ is concave in the interval $(0,\omega_{ij}-\ji)$, as desired. 
\end{proof}

In other words, given the choices of all players except 
player $i$,  the utility of 
the $i$-th player from the $j$-th CPR is a  concave function, when restricted on a particular interval. 
The next result shows 
that an analogous statement holds true for the total 
utility of each player in the Fragile multi-CPR Game, namely, $\mV_i(\xb_i;\xb_{-i})$, given by~\eqref{utility}. 

Given $i\in [n]$ and $\xb_{-i}\in \mC_{-i}$, 
let 
\begin{equation}\label{active_cpr}
A(\xb_{-i}):=\{j\in [m]:\ji < \omega_{ij} \}\, , 
\end{equation}
where $\omega_{ij}, j\in [m]$, is provided by  Lemma~\ref{VTP_lemma}. 
We refer to $A(\xb_{-i})$ as the set of \emph{active CPRs} corresponding to $i$ and $\xb_{-i}$. 

\begin{theorem}\label{concave}
Fix $i\in [n]$  and $\xb_{-i}\in\mC_{-i}$. 
Let $A(\xb_{-i})$ be the set of active CPRs corresponding to $i$ and $\xb_{-i}$,  and consider the set $\mR_{A(\xb_{-i})} = \prod_{j\in A(\xb_{-i})} (0,\omega_{ij} - \ji)$. 
Then the function $\mV_{A(\xb_{-i})}:=\sum_{j\in A(\xb_{-i})}\mE_{ij}(x_{ij}; \xt^{j\setminus i})$  is concave in $\mR_{A(\xb_{-i})}$.  
\end{theorem}
\begin{proof} 
If $|A(\xb_{-i})|=1$, then the result follows from Lemma~\ref{VTP_lemma} so we may assume that $|A(\xb_{-i})|\ge 2$. 
The set $\mR_{A(\xb_{-i})}$ is clearly convex. 
Let $j,k\in A(\xb_{-i})$ be such that $j\neq k$ and 
notice that 
\begin{equation}\label{jk}
\frac{\partial^2 \mV_{A(\xb_{-i})}}{\partial x_{ij} \; \partial x_{ik}} = 0 \, .
\end{equation}
Moreover, by Lemma~\ref{VTP_lemma}, we also have   
\begin{equation}\label{jj}
\frac{\partial^2 \mV_{A(\xb_{-i})}}{\partial x_{ij}^2} = \frac{\partial^2 \mE_{ij}(x_{ij};\ji)}{\partial x_{ij}^2} < 0 \, , \text{ for all } x_{ij}\in (0,\omega_{ij}-\ji) \, .
\end{equation}
Given $\xb\in\mR_{A(\xb_{-i})}$,  denote by  $H(\xb) = \left( \frac{\partial^2 \mV_{A(\xb_{-i})}(\xb)}{\partial x_{ij} \; \partial x_{ik}}\right)_{j,k\in A(\xb_{-i})}$ the Hessian matrix of $\mV_{A(\xb_{-i})}$ evaluated at $\xb$,  and let $\Delta_k(\xb)$, for $k\in A(\xb_{-i})$, be the principal minors of $H(\xb)$ (see~\cite[p.~111]{Berkovitz}).  Notice that~\eqref{jk} implies that $H(\xb)$ is a diagonal matrix. Therefore, using~\eqref{jj}, it follows that 
$(-1)^k\cdot\Delta_k(\xb) >0$, when $\xb\in \mR_{A(\xb_{-i})}$. In other words, $H(\,\cdot\,)$ is negative definite on the convex set $\mR_{A(\xb_{-i})}$  and we conclude  (see~\cite[Theorem~3.3, p.~110]{Berkovitz})  that  
$\mV_{A(\xb_{-i})}$ is concave in $\mR_{A(\xb_{-i})}$, as desired.
\end{proof}

\section{Proof of Theorem~\ref{unique_NE}: existence of GNE}\label{sec:3}

In this section we show that the Fragile multi-CPR Game possesses a Generalized Nash equilibrium. 
Recall that the notion of Generalized Nash equilibrium 
depends upon the choice of constraint policies. 
Thus, before presenting the details of the proof, 
we first define the constrained policies under consideration.

Let $i\in [n]$ and 
$\xb_{-i}=(\xb_1,\ldots,\xb_{i-1},\xb_{i+1},\ldots,\xb_{n})\in \mC_{-i}$, where $\xb_j = (x_{j1},\ldots,x_{jm})\in C_m$, for $j\in [n]\setminus \{i\}$, be fixed. 
Recall that 
$\ji = \sum_{\ell\in [n]\setminus  \{i\}} x_{\ell j}$ and consider the set of active indices corresponding to $i$ and $\xb_{-i}$, i.e., consider the set
$A(\xb_{-i})$, defined in~\eqref{active_cpr}.

Define the constraint policy  $\vartheta_i(\cdot)$ 
that maps each element $\xb_{-i}\in\mC_{-i}$ to the set 
\begin{equation}\label{def:constraint_policy}
\vartheta_i(\xb_{-i}) =C_m \bigcap \left\{ \prod_{j\in A(\xb_{-i}) } [0,\omega_{ij} - \ji ]\,\, \times \prod_{j\in [m]\setminus A(\xb_{-i}) } \{0\} \right\}\, ,
\end{equation}
where $\{\omega_{ij}\}_{j\in A(\xb_{-i})}$ is given by Lemma~\ref{VTP_lemma}. 
Notice that, for every $\xb_{-i}\in\mC_{-i}$, the set $\vartheta_i(\xb_{-i})$ is  \emph{non-empty, compact and convex}. Fig. \ref{fig:1} provides a visualization of the aforementioned constraint policy, in the case of $m=2$.

%Figure for the visualization of the constraint set
%\begin{figure}[!hbtp] 
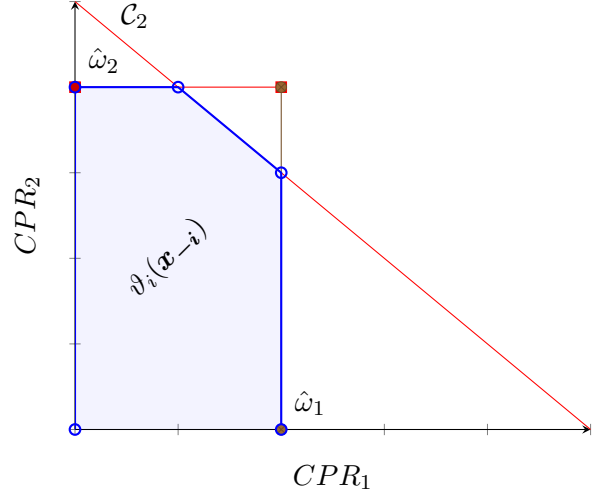
\begin{figure}[t] 
\begin{center}
\begin{tikzpicture}
\begin{axis}[
    axis lines = left,
    xlabel = $CPR_1$,
    ylabel = {$CPR_2$},
    yticklabels={,,},
    xticklabels={,,}
]
\addplot [
    domain=0:1, 
    samples=100, 
    color=red,
]
{-x+1};

\addplot coordinates {(0, 0.8) (0.4, 0.8)};
\addplot coordinates {(0.4, 0.8) (0.4, 0)};

%it is not exactly omega here...
\node[label={$\quad \quad \hat{\omega}_2$},circle,fill,inner sep=1pt] at (axis cs:0,0.8) {};
\node[label={$\quad \quad \hat{\omega}_1$},circle,fill,inner sep=1pt] at (axis cs:0.4,0) {};
\node[label={$\quad \quad \mathcal{C}_2$}] at (axis cs:0.06,0.89) {};
\addplot[thick,color=blue,mark=o,fill=blue, 
                    fill opacity=0.05]coordinates  {
            (0, 0) 
            (0, 0.8)
            (0.2, 0.8)
            (0.4, 0.6)
            (
            (0.4, 0)
            };
            \node [rotate=45] at (axis cs:  .18,  .40) {$\vartheta_i(\bm{x_{-i}})$};
            
\end{axis}
\end{tikzpicture}

\end{center}
\caption{Visualization of an instance of the constraint policy $\vartheta_i(\cdot)$ (blue shaded region) of player $i$ in the case of $m=2$,  where we denote  $\hat{\omega}_1 = \omega_{i1} - \mathbf{x}_T^{1|i} $ and  $\hat{\omega}_2 = \omega_{i2} - \mathbf{x}_T^{2|i} $.}
\vspace{0.7cm}
\label{fig:1}
\end{figure}

We aim to show that the Constrained Fragile multi-CPR Game, with constraint policies given by~\eqref{def:constraint_policy}, admits a Generalized Nash equilibrium. In order to do so, we employ the following theorem.  Recall (see~\cite[p.~32--33]{Ichiishi}) that a set-valued correspondence $\phi:X\to 2^Y$ is \emph{upper semicontinuous} if for every open set $G\subset Y$, it holds that  $\{x\in X : \phi(x) \subset G\}$ is an open set in $X$. A set-valued correspondence $\phi:X\to 2^Y$ is \emph{lower semicontinuous} if every open set $G\subset Y$, it holds that $\{x\in X : \phi(x) \cap G \neq \emptyset \}$ is an open set in $X$. Recall also that, given $S\subset \mathbb{R}^s$, a function $f:S\to\mathbb{R}$ is quasi-concave if $f(\lambda \xb + (1-\lambda)\yb) \ge \min\{f(\xb),f(\yb)\}$, for all $\xb\neq\yb$ in $S$ and $\lambda\in (0,1)$. Clearly, a concave function is also quasi-concave.

\begin{theorem}\label{thm:gne}
Let $n$ players be characterized by strategy spaces $X_i, i\in [n]$,  constraint policies  $\phi_i, i\in [n]$, and  utility functions $\mV_i: \prod_i X_i \to \mathbb{R}, i\in [n]$. Suppose further that the following hold true for every $i\in [n]$:
\begin{enumerate}
    \item $X_i$ is non-empty, compact, convex subset of a Euclidean space. 
    \item $\phi_i(\cdot)$ is both upper semicontinuous and lower semicontinuous in $X_{-i}$. 
    \item For all $\xb_{-i}\in X_{-i}$, $\phi_i(\xb_{-i})$ is  nonempty, closed and convex. 
    \item $\mV_i$ is continuous in $\prod_i X_i$. 
    \item For every $\xb_{-i}\in X_{-i}$, the map $x_i\mapsto \mV_i(x_i,\xb_{-i})$ is quasi-concave on $\phi_i(\xb_{-i})$. 
\end{enumerate}
Then there exists a Generalized Nash equilibrium. 
\end{theorem}
\begin{proof}
This is a folklore result that can be found in various places. See, for example, \cite{Arrow_Debreu}, \cite[Theorem~6]{Facchinei_Kanzow},   \cite[Theorem~4.3.1]{Ichiishi},  \cite[Theorem~12.3]{Aubin},  or \cite[Theorem~3.1]{Dutang}.  
\end{proof}

We are now ready to establish the existence of a GNE in the Constrained Fragile multi-CPR Game. 
In the following proof, $\|\cdot\|_d$ denotes $d$-dimensional Euclidean distance, and 
$B_d(\eps):= \{\xb\in \mathbb{R}^d: \|\xb\|_d \le \eps\}$ is the closed ball of radius $\eps$ centered at the origin. 
Moreover, given $A\subset \mathbb{R}^d$ and $\eps>0$, we denote by $\{A\}_{\eps}$  the set $A + B_d(\eps):=\{a + b: a\in A \text{ and } b\in B_d(\eps)\}$ and by $(1-\eps)\cdot A$ the set 
$\{(1-\eps)\cdot a : a\in A\}$.

\begin{proof}[Proof of Theorem~\ref{unique_NE}]
We apply Theorem~\ref{thm:gne}. The strategy space of each player is equal to $C_m$, which is non-empty, compact and convex. Hence the first condition of Theorem~\ref{thm:gne} holds true. 
The third condition also holds true, by~\eqref{def:constraint_policy}. Moreover, the fourth condition of Theorem~\ref{thm:gne} is immediate from the definition of utility, given in~\eqref{utility}, while the  fifth condition follows from Theorem~\ref{concave}.

It remains to show that the second condition of Theorem~\ref{thm:gne} holds true, i.e., that for each $i\in [n]$ the constrained policy $\vartheta_i(\cdot)$, given by~\eqref{def:constraint_policy}, is both upper and lower semicontinuous.  Towards this end, 
fix $i\in [n]$ and let $G\subset C_m$ be an open set. 
Consider the sets 
\[
G^{+}:=\{\xb_{-i}\in \mC_{-i}: \vartheta_i(\xb_{-i}) \subset G\} \quad \text{and} \quad 
G^{-}:=\{\xb_{-i}\in \mC_{-i}: \vartheta_i(\xb_{-i}) \cap G \neq \emptyset\} \, .
\]
We have to show that both $G^+$ and $G^-$ are 
open subsets of $\mC_{-i}$. 
We first show that $G^+$ is open. 

If $G^+$ is empty then the result is clearly true, so we may assume that 
$G^+\neq \emptyset$. Let $\yb=(\yb_1,\ldots,\yb_{i-1},\yb_{i+1},\ldots,\yb_n) \in G^+$; hence $\vartheta_i(\yb) \subset G$. We have to show that there exists $\eps>0$ such that for every $\xb\in\mC_{-i}$ with  $\|\xb-\yb\|_{(n-1)m}<\eps$, we have $\vartheta_i(\xb)\subset G$. 
Since $\vartheta_i(\yb)$ is a compact subset of the open set $G$, it follows that there exists $\eps_0>0$ such that 
$\{\vartheta_i(\yb)\}_{\eps_0} \subset G$. 
Since summation is continuous, 
there exists $\eps_1>0$ such that for every 
$\xb\in \mC_{-i}$ with
$\|\xb-\yb\|_{(n-1)m}<\eps_1$ it holds   
$\xb \in \{\vartheta_i(\yb)\}_{\eps_0}$. 
The desired $\eps$ is given by $\eps_1$. 
Hence $G^+$ is an open set, and we proceed with showing 
that $G^-$ is open as well. 

We may assume that $G^-$ is non-empty. 
For each $i\in [n]$, let $g_i: \mC_{-i} \to \mathbb{R}^m_{\geq 0}$ be the continuous function 
whose $j$-th coordinate, for $j\in [m]$, is given by 
\begin{equation*}
    g_{ij}(\xb_{-i}) = \begin{cases}  \omega_{ij} - \ji \, , & \mbox{ if } \omega_{ij} - \ji > 0 \\
    0 \, , & \mbox{ if } \omega_{ij} - \ji \leq 0 \,\, ,
    \end{cases}
\end{equation*}
where $\omega_{ij}$ is given by Lemma~\ref{VTP_lemma}. 
Let $h:\mathbb{R}^m_{\geq 0} \to 2^{C_m}$ be the set-valued function defined by 
$h(z_1,\ldots,z_m)= \prod_{j\in [m]} [0,z_j]$, 
with the convention $[0,0]:=\{0\}$. Clearly, it holds that $\vartheta_i = h \circ g_i$, for all $i\in [n]$. 

We claim that $h$ is lower semicontinuous. 
If the claim holds true then it follows that the set  $H := \{ \zb \in \mathbb{R}^m_{\geq 0}: h(\zb) \cap G \neq \emptyset\}$ is open. Notice that 
$G^-\neq\emptyset$ implies that $H\neq\emptyset$. 
Since $g_i$ is continuous, it follows that the 
preimage of $H$ under $g_i$, i.e., $g^{-1}_i(H)$, 
is open. In other words, the set 
$\{\xb\in\mC_{-i} : h\circ g_i(\xb) \cap G\neq\emptyset \}= \{\xb\in\mC_{-i} : \vartheta_i(\xb) \cap G\neq\emptyset \}$ 
is open and the proof of the theorem is complete. 

It remains to prove the claim, i.e., that 
$h$ is lower semicontinuous. To this end, 
let $G \subset C_m$ be an open set, and let 
$G^{\ast}:=\{\zb \in \mathbb{R}^m_{\geq 0}: h(\zb) \cap G \neq \emptyset \}$. 
We have to show that $G^{\ast}$ is open; that is, 
we have to show that for every $\zb\in G^{\ast}$ there exists $\eps>0$ such that  $\wb\in G^{\ast}$, for 
all $\wb$ with $\|\zb-\wb\|_m<\eps$.  Fix $\zb\in G^{\ast}$. 
Since $h(\zb)$ 
is compact and $G$ is open, it follows that there 
exists $\eps_0>0$ such that $(1-\eps_0)\cdot h(\zb) \cap G\neq \emptyset$. Now choose $\eps>0$ such that 
for every $\wb\in C_m$ for which $\|\zb-\wb\|_m <\eps$ it holds $(1-\eps_0)\cdot h(\zb)\subset h(\wb)$. In other words, for this particular choice of $\eps>0$   it holds $h(\wb)\cap G\neq\emptyset$, for every $\wb$ with $\|\zb-\wb\|_m <\eps$. The claim follows. 
\end{proof}

\section{Best response correspondence}\label{sec:4}

Having established the existence of a GNE for 
a Fragile multi-CPR Game, we now proceed with the proofs of Theorems~\ref{one_CPR} and~\ref{nr_of_GNE}. The proofs will be obtained in two steps.  
In the first step we deduce certain 
``first order conditions" which are satisfied by the best response 
correspondence of each player in the game. 
In the second step we employ the first order conditions in order to define certain 
auxiliary functions, whose monotonicity 
will be employed in the proofs of the aforementioned theorems. In this section we collect some results pertaining to the first step. 
We begin with recalling the notion of the best response correspondence (see~\cite{Laraki_et_al}).

Given $i\in [n]$ and 
$\xb_{-i}\in\mC_{-i}$, let $\vartheta_i(\cdot)$ denote the constraint policy given by~\eqref{def:constraint_policy}, and consider the \emph{best response} of the $i$-th player in the Fragile multi-CPR Game defined as follows:
\begin{equation}\label{def_BR}
    B_i(\xb_{-i}) = \arg\max_{\xb_i\in\vartheta_i(\xb_{-i})}\, \mV_i(\xb_i; \xb_{-i}) \, ,
\end{equation}
where $\mV_i$ is the utility of the $i$-th player,  given by~\eqref{utility}.
Notice that $B_i(\cdot)$ is a correspondence $B_i:\mC_{-i}\to 2^{C_m}$, where 
$2^{C_m}$ denotes the class consisting of all subsets of $C_m$. For $j\in [m]$, we denote by $B_{ij}(\xb_{-i})$ the $j$-th component of $B_i(\xb_{-i})$; hence we have \[
B_i(\xb_{-i}) = (B_{i1}(\xb_{-i}),\ldots, B_{im}(\xb_{-i}))\, .
\]

\begin{remark}\label{br_NE}
Notice that  Definition~\ref{def:gne} implies that if 
$\xb=(\xb_1,\ldots,\xb_n)\in\mC_n$ is a GNE of a Constrained Fragile multi-CPR Game, with constraint policies given by~\eqref{def:constraint_policy}, then for each $i\in [n]$ it holds $\xb_{i}\in B_{i}(\xb_{-i})$. 
\end{remark}

Recall that  
$A(\xb_{-i})$
denotes the set of active CPRs corresponding to 
$\xb_{-i}$, defined in~\eqref{active_cpr}, and notice that $B_{ij}(\xb_{-i})=0$, for all $j\in [m]\setminus A(\xb_{-i})$. 

For $x_{ij}\in [0,1]$, let $\psi_{ij}(x_{ij} ; \ji)$ be the function defined via 
\begin{equation}\label{psi_function}
\psi_{ij}(x_{ij}; \ji) = x_{ij}\cdot \frac{\partial }{\partial x_{ij}}\mF_{ij}(x_{ij} +\ji) + a_i \mF_{ij}(x_{ij} +\ji) \, .
\end{equation}

\begin{lemma}\label{BR_thm}
Fix $i\in [n]$ and $\xb_{-i}\in\mC_{-i}$  and let $\mR_{A(\xb_{-i})} = \prod_{j\in A(\xb_{-i})}(0,\omega_{ij}- \ji)$, where $\omega_{ij}$ is provided by Lemma~\ref{VTP_lemma}. 
Then a global maximum of the  function  
$\mV_{\xb_{-i}} := \sum_{j\in A(\xb_{-i})}\mE_{ij}(x_{ij} ; \ji)$ defined on the set $\mR_{A(\xb_{-i})}$ is given by the unique solution of the following system of equations:
\begin{equation}\label{eq:system}
\psi_{ij}(x_{ij}; \ji) = 0 , \,  \text{ for } \, j\in A(\xb_{-i}) \, ,
\end{equation}
where $\psi_{ij}(x_{ij}; \ji)$ is defined in~\eqref{psi_function}. 
\end{lemma}
\begin{proof} 
To simplify notation, we write $\psi_{ij}(\cdot)$ instead of $\psi_{ij}(\,\cdot\, ; \ji)$. 
Using~\eqref{1_dimensional} and~\eqref{utility}, it is straightforward to verify that for every $j\in A(\xb_{-i})$ it holds 
\begin{equation}\label{derivative_utility}
\frac{\partial \mV_{\xb_{-i}}}{\partial x_{ij}} =  \frac{\partial \mE_{ij}(x_{ij}; \ji )}{\partial x_{ij}}  = x_{ij}^{a_i-1} \cdot \psi_{ij}(x_{ij}) \, .
\end{equation}
Now notice that $\psi_{ij}(0) >0$ as well as  
$\psi_{ij}(\omega_{ij}-\ji) <0$. Moreover, Assumption~\ref{ass_prob} readily implies that $\psi_{ij}(\cdot)$ is strictly decreasing on the interval  $(0,\omega_{ij}-\ji)$. 
The intermediate value theorem implies that there exists unique $\lambda_{ij}\in (0,\omega_{ij}-\ji)$ such that 
$\psi_{ij}(\lambda_{ij}) = 0$. 
Hence, it follows from~\eqref{derivative_utility} that the points $\lambda_{ij}$, for $j\in A(\xb_{-i})$, are  critical points of the  function $\mV_{\xb_{-i}}$, which is concave on the open and convex set $\mR_{A(\xb_{-i})}$, 
by Theorem~\ref{concave}. 
It follows (see~\cite[Theorem~2.4, p.~132]{Berkovitz}) that 
$\{\lambda_{ij}\}_{j\in A(\xb_{-i})}$
is a global maximum of $\mV_{\xb_{-i}}$ 
on $\mR_{A(\xb_{-i})}$. 
We conclude that $\mV_{\xb_{-i}}$ is maximized when $x_{ij}=\lambda_{ij}$, for $j\in A(\xb_{-i})$, as desired.
\end{proof}

\begin{remark}\label{rem:1}
Let us remark that the solution of the system of equations given by~\eqref{eq:system} may not belong to the set $C_m$. More precisely, it could 
happen that the solution of the system of equations~\eqref{eq:system}, say $\{\lambda_{ij}\}_{j\in A(\xb_{-i})}$, satisfies 
$\sum_{j\in A(\xb_{-i})} \lambda_{ij} >1$. 
This is a crucial difference between the Fragile CPR Game and the Fragile multi-CPR Game. 
\end{remark}

Now notice that, given $\xb_{-i}\in\mC_{-i}$, the best response of player $i$ is a local maximum of the following program:

\begin{equation*}%\label{opt_problem2.0}
\begin{aligned}
& \underset{\{x_{ij}\}_{j\in A(\xb_{-i})} }{\text{maximize}}
& & \mV_{\xb_{-i}} := \sum_{j\in A(\xb_{-i}) }\mE_{ij}(x_{ij} ; \ji)\\
& \text{subject to}
& & \sum_{j\in A(\xb_{-i}) } x_{ij}\le 1   \\
& & & 0\le  x_{ij} \le \omega_{ij} - \ji, \text{ for all } \, j \in A(\xb_{-i}) \, .
\end{aligned}
\end{equation*}
Equivalently, the best response of player $i$ is a local minimum of the following program:

\begin{equation}\label{opt_problem2}
\begin{aligned}
& \underset{\{x_{ij}\}_{j\in A(\xb_{-i})} }{\text{minimize}}
& &   -\sum_{j\in A(\xb_{-i}) }\mE_{ij}(x_{ij} ; \ji)\\
& \text{subject to}
& & \sum_{j\in A(\xb_{-i}) } x_{ij}\le 1   \\
& & & 0\le  x_{ij} \le \omega_{ij} - \ji, \text{ for all } \, j \in A(\xb_{-i}) \, .
\end{aligned}
\end{equation}

Notice that since $\mE_{ij}(\,\cdot\, ; \ji)$ is concave on $(0,\omega_{ij} - \ji)$, by Lemma~\ref{VTP_lemma},  it follows that  Problem~\eqref{opt_problem2} is a separable convex knapsack program (see~\cite{Levi_Perakis_Romero, Stefanov}). We are going to describe the  optima of 
Problem~\eqref{opt_problem2}  using the KKT conditions. 
The KKT conditions pertain to the Lagrangian corresponding to Problem~\eqref{opt_problem2}, which is defined as the following quantity:
\[
\mathcal{L} := -\mV_{\xb_{-i}} +\kappa_0 \cdot\left(\sum_{j\in A(\xb_{-i}) } x_{ij}- 1\right) + \sum_{j\in A(\xb_{-i}) }\mu_j\cdot (x_{ij} + \ji -\omega_{ij}) + \sum_{j\in A(\xb_{-i}) }\nu_j\cdot (-x_{ij})  ,
\]
where $\kappa_0, \{\mu_j\}_j, \{\nu_j\}_j$ are real numbers. 
The KKT conditions corresponding to problem~\eqref{opt_problem2} read as follows (see~\cite[Theorem~3.8]{Luptacik}). 

\begin{theorem}[KKT conditions for Problem~\eqref{opt_problem2}]\label{kkt_cond} 
If $\{x_{ij}\}_{j\in A(\xb_{-i}) }$ is a local minimum of Problem~\eqref{opt_problem2}, then there exist \emph{non-negative} real numbers $\kappa_0$, $\{\mu_j\}_{j\in A(\xb_{-i}) }$, and $\{\nu_j\}_{j\in A(\xb_{-i}) }$ such that:
\begin{enumerate}
\item For all $j\in A(\xb_{-i})$ it holds $-x_{ij}^{a_i-1} \cdot \psi_{ij}(x_{ij};\ji) + \kappa_0 + \mu_j - \nu_j =0\,$, where $\psi_{ij}$ is given by~\eqref{psi_function}. 
\item $\kappa_0 \cdot\left(\sum_{j\in A(\xb_{-i})} x_{ij}- 1\right)=0$. 
\item $\mu_j\cdot (x_{ij} + \ji -\omega_{ij}) =0\,$, for all $j\in A(\xb_{-i})$. 
\item $\nu_j \cdot x_{ij}=0\,$, for all $j\in A(\xb_{-i})$. 
\item $0\le  x_{ij} \le \omega_{ij} - \ji\,$, for all   $j \in A(\xb_{-i})$. 
\end{enumerate}
\end{theorem}

We aim to employ Theorem~\ref{kkt_cond} in order to 
describe a local maximum of  Problem~\eqref{opt_problem2} via the solution of a system of equations. 
This will require the following result, which is presumably reported somewhere in the literature but, lacking a reference, we include a proof for the sake of completeness. 

\begin{lemma}\label{opt_lemma}
Fix a positive integer $s$ and, for each $j\in [s]$, let  $f_j:\mathbb{R}\to \mathbb{R}$ be a strictly decreasing function. Then there exists at most one vector $(c,x_1,\ldots,x_s)\in \mathbb{R}^{s+1}$ such that 
\[
f_j(x_j) = c,\, \text{ for all } \, j\in [s], \, \text{ and } \, \sum_{j\in [s]} x_j =1 \, .
\]
\end{lemma}
\begin{proof}
Suppose that there exist two distinct  vectors, say 
$(c,x_1,\ldots,x_s)$ and $(d,y_1,\ldots, y_s)$. 
If $c=d$, then there exists $j\in [s]$ such that $x_j\neq y_j$ and  
\[
f_j(x_j) = c = d = f_j(y_j) \, ,
\]
contrariwise to the assumption that the function  $f_i(\cdot)$ is strictly decreasing. Hence $c\neq d$. 

Since $f_j(\cdot), j\in [s]$, is strictly decreasing, it is injective and therefore it follows that it is invertible. Let us denote its inverse by $f_j^{-1}(\cdot)$. We then have 
\[
x_j = f_j^{-1}(c) \, \text{ and } y_j = f_j^{-1}(d), \, \text{ for all } \, j\in [s], 
\]
which in turn implies that $x_j \neq y_j$, for all $j\in [s]$. Assume, without loss of generality, that 
$c < d$. The assumption that $f_j$ is strictly decreasing then implies  
$x_j > y_j$,  for all  $j\in [s]$,
and therefore 
$1 = \sum_{j\in [s]} x_j > \sum_{j\in [s]} y_j =1$,
a contradiction. The result follows. 
\end{proof}

We may now proceed with describing the 
best responses of each player in the 
Fragile multi-CPR Game via a system of ``first order conditions". 

\begin{theorem}\label{kkt_thm} 
Let $i\in [n]$ and $\xb_{-i}\in\mC_{-i}$ be fixed. Suppose that $\{x_{ij}\}_{j\in A(\xb_{-i})}$ is a best response of player $i$ in the Fragile multi-CPR Game. Then 
$\{x_{ij}\}_{j\in A(\xb_{-i})}$ is either  of the following two types:
\begin{itemize}
\item \textbf{Type~I: } There exists $J_{\xb_{-i}}\subset A(\xb_{-i})$ such that 
$x_{ij}=0$, when $j\in A(\xb_{-i})\setminus J_{\xb_{-i}}$, and  $\{x_{ij}\}_{j\in J_{\xb_{-i}}}$ satisfy the following inequality, and are given by the unique solution of 
the following system of equations: 
\[
\sum_{j\in J_{\xb_{-i}}}x_{ij} < 1 \quad \text{ and } \quad  \psi_{ij}(x_{ij}; \ji) = 0, \, \text{ for }\,  j\in J_{\xb_{-i}}\, ,
\]
where $\psi_{ij}(\,\cdot\, ; \ji)$ is defined in~\eqref{psi_function}.
\item \textbf{Type~II: } There exists $J_{\xb_{-i}}\subset A(\xb_{-i})$ and a real number $\kappa_0 \ge 0$ such that 
$x_{ij}=0$, when $j\in A(\xb_{-i})\setminus J_{\xb_{-i}}$, and   $\{x_{ij}\}_{j\in J}$ are given by the unique solution of 
the following system of equations: 
\[
\sum_{j\in J_{\xb_{-i}}} x_{ij} =1 \quad \text{and} \quad 
x_{ij}^{a_i-1} \cdot \psi_{ij}(x_{ij}; \ji) = \kappa_0, \, \text{ for } \,  j\in J_{\xb_{-i}} \, , 
\]
where $\psi_{ij}(\,\cdot\, ; \ji)$ is defined in~\eqref{psi_function}.
\end{itemize}
\end{theorem}
\begin{proof} 
Let $\{x_{ij}\}_{j\in A(\xb_{-i}) }$ be
a best response of player $i\in [n]$. Then 
$\{x_{ij}\}_{j\in A(\xb_{-i}) }$ is a local minimum
of Problem~\eqref{opt_problem2}; hence it   satisfies the KKT Conditions of Theorem~\ref{kkt_cond}.

If $x_{ij} = \omega_{ij} - \ji$, for some $j\in A(\xb_{-i})$, then  Lemma~\ref{VTP_lemma} 
and~\eqref{1_dimensional} imply that 
$\mE_{ij}(x_{ij};\ji)=0$. Hence player $i$ could achieve the 
same utility from the $j$-th CPR by choosing $x_{ij}=0$. Thus we may assume that $x_{ij} <\omega_{ij}- \ji$, for all 
$j\in A(\xb_{-i})$ and therefore  Theorem~\ref{kkt_cond}.(3) implies that $\mu_j =0$, for all 
$j\in A(\xb_{-i})$. 
Now let 
\begin{equation}\label{J_set}
J_{\xb_{-i}}=\{j\in A(\xb_{-i}): x_{ij}\neq 0\} \, ,
\end{equation}
and notice that Theorem~\ref{kkt_cond}.(4) implies that $\nu_j=0$ for $j\in J_{\xb_{-i}}$.
We distinguish two cases. 

Suppose first that $\sum_{j\in J_{\xb_{-i}}}x_{ij} < 1$. 
Then Theorem~\ref{kkt_cond}.(2) yields  $\kappa_0=0$, and therefore   Theorem~\ref{kkt_cond}.(1) implies   
that  $\{x_{ij}\}_{i\in J_{\xb_{-i}}}$ is given by  the unique solution of the following  system of equations:
\[
 \psi_{ij}(x_{ij}; \ji) = 0, \, \text{ for } \,  j\in J_{\xb_{-i}} \, .
\]
In other words, if $\sum_{j\in J_{\xb_{-i}}}x_{ij} < 1$ then $\{x_{ij}\}_{j\in A(\xb_{-i})}$ is of Type~I. 

Now assume that $\sum_{j\in J_{\xb_{-i}}}x_{ij} = 1$. 
Then  Theorems~\ref{kkt_cond}.(1) and~\ref{kkt_cond} .(2) imply that there exists $\kappa_0 \ge0$ such that $-x_{ij}^{a_i-1} \cdot \psi_{ij}(x_{ij};\ji) = -\kappa_0$, for all $j\in J_{\xb_{-i}}$. In other words,  
$\{x_{ij}\}_{j\in J_{\xb_{-i}}}$  and $\kappa_0$ are given by the solution of the following system of equations:
\begin{equation}\label{equal_one}
\sum_{j\in J_{\xb_{-i}}} x_{ij} =1\, \text{ and } \, 
x_{ij}^{a_i-1} \cdot \psi_{ij}(x_{ij}; \ji) = \kappa_0, \, \text{ for all } \,  j\in J_{\xb_{-i}} \, .
\end{equation}
Since the functions $f_{ij}(x_{ij}):= x_{ij}^{a_i-1}\cdot\psi_{ij}(x_{ij} ; \ji)$, for $j\in J_{\xb_{-i}}$, are 
strictly decreasing, Lemma~\ref{opt_lemma} implies that the 
system of equations  in~\eqref{equal_one} has a
unique solution. 
Hence $\{x_{ij}\}_{j\in A(\xb_{-i})}$ is of Type~II and 
the result follows. 
\end{proof}

We refer to the set $J_{\xb_{-i}}$ provided by 
Theorem~\ref{kkt_thm}, defined in~\eqref{J_set}, as the set of \emph{effective} CPRs corresponding to $i\in [n]$ and $\xb_{-i}\in\mC_{-i}$. 
In the next section we employ Theorem~\ref{kkt_thm} 
in order to define auxiliary functions (i.e.,~\eqref{g_funct} and~\eqref{h_funct} below) whose monotonicity will play a key role in the proof of 
Theorem~\ref{nr_of_GNE}.

\section{Auxiliary functions}\label{sec:5}

In this section  we define and state 
basic properties of certain auxiliary functions, whose monotonicity will be used in the proofs of Theorems~\ref{one_CPR} and~\ref{nr_of_GNE}, and whose definition depends upon the ``first order conditions" provided by Theorem~\ref{kkt_thm}.

Let us begin with some notation and remarks.  
Fix $i\in [n]$ and $\xb_{-i}\in\mC_{-i}$, and recall from~\eqref{def_BR} that $B_i(\xb_{-i})$ denotes a best response of player $i$ and that  $B_{ij}(\xb_{-i})$ is its $j$-th component. 
To simplify notation, 
let us denote $b_{ij}:= B_{ij}(\xb_{-i})$. 
From Theorem~\ref{kkt_thm} we know that  there exists $J_{\xb_{-i}}\subset A(\xb_{-i})$ such that $b_{ij}=0$, for $j\in A(\xb_{-i})\setminus J_{\xb_{-i}}$, and either 
\begin{equation}\label{psi1}
\sum_{j\in J_{\xb_{-i}}} b_{ij} < 1 \quad \text{and} \quad
\psi_{ij}(b_{ij};\xb_{-i}) = 0, \,\text{ for all } \, j\in J_{\xb_{-i}}, 
\end{equation} 
or    
\begin{equation}\label{psi2}
\sum_{j\in J_{\xb_{-i}}}b_{ij}=1 \quad \text{ and } \quad b_{ij}^{a_i-1}\cdot \psi_{ij}(b_{ij};\xb_{-i}) = \kappa_0, \,\text{ for all } \, j\in J_{\xb_{-i}}\, \text{ and some } \, \kappa_0\ge 0.
\end{equation}
In particular, it holds $b_{ij}>0$, for all $j\in J_{\xb_{-i}}$. 
Using~\eqref{psi_function}, it follows that the second statement of~\eqref{psi1} is equivalent to 
\begin{equation}\label{psi3}
 b_{ij}\cdot \frac{\partial}{\partial x_{ij}} \mF_{ij}(b_{ij} +\ji) + a_i \mF_{ij}(b_{ij} +\ji) = 0 , \,\text{ for all } \, j\in J_{\xb_{-i}},
\end{equation}
and that the second statement of~\eqref{psi2} is equivalent to
\begin{equation}\label{psi4}
b_{ij}^{a_i-1}\cdot \left( b_{ij}\cdot \frac{\partial}{\partial x_{ij}} \mF_{ij}(b_{ij} +\ji) + a_i \mF_{ij}(b_{ij} +\ji) \right) = \kappa_0, \,\text{ for all } \, j\in J_{\xb_{-i}}\, . 
\end{equation} 
Now, given $\xb_{-i}\in\mC_{-i}$, $j\in J_{\xb_{-i}}$ and $\kappa_0\ge 0$, define 
for each $i\in [n]$  the   functions 
\begin{equation}\label{g_funct}
\mG_{ij}(x_{ij}+\ji) := -\frac{ a_i \mF_{ij}(x_{ij} +\ji)}{\frac{\partial}{\partial x_{ij} } \mF_{ij}(x_{ij} +\ji)} \, , \text{ for } \, x_{ij}\in (0,\omega_{ij}-\ji)
\end{equation}
and 
\begin{equation}\label{h_funct}
\mH_{ij}(x_{ij}+\ji ; \kappa_0) := -\frac{ a_i \mF_{ij}(x_{ij} +\ji)} { \frac{-\kappa_0}{x_{ij}^{a_i}} +\frac{\partial}{\partial x_{ij}} \mF_{ij}(x_{ij} +\ji)} , \text{ for } \, x_{ij}\in (0,\omega_{ij}-\ji)\, .
\end{equation}
Notice that~\eqref{psi3} implies that when $b_{ij}$ is of Type~I it holds 

\begin{equation}\label{fixed_g}
\mG_{ij}( b_{ij}+\ji) = b_{ij} \, ,
\end{equation}
while~\eqref{psi4} implies that when $b_{ij}$ is of Type~II it holds   
\begin{equation}\label{fixed_h}
\mH_{ij}( b_{ij} +\ji; \kappa_0) = b_{ij}, \, .
\end{equation}
Observe also that it holds 
$\mG_{ij}(x_{ij}+\ji) \ge \mH_{ij}(x_{ij}+\ji; \kappa_0)$, for all $x_{ij}\in [0,\omega_{ij}-\ji]$. 
Let us, for future reference, collect a couple of observations about the functions $\mG_{ij},\mH_{ij}$. 

\begin{lemma}\label{gh_functions}
Let $i\in [n]$ and $j\in [m]$ be fixed. 
Then the functions $\mG_{ij}(\cdot)$ and $\mH_{ij}(\,\cdot\, ; \kappa_0)$, defined in~\eqref{g_funct} and~\eqref{h_funct} respectively, are strictly decreasing in the interval $[0,\omega_{ij}]$.  
\end{lemma}
\begin{proof}
To simplify notation, let $\mF:=\mF_{ij}(x_{ij} + \ji)$,   $\mF^{\prime}:=\frac{\partial}{\partial x_{ij}} \mF$ and $\mF^{\prime\prime}:=\frac{\partial^2}{\partial x_{ij}^2} \mF$. For $x_{ij}\in (0,\omega_{ij}-\xtj)$, 
we compute 
\begin{eqnarray*}
\frac{\partial}{\partial x_{ij}}\mH_{ij}(x_{ij}+\ji; \kappa_0) &=&  \frac{ -a_i \mF^{\prime} \cdot (\frac{-\kappa_0}{x_{ij}^{a_i}} +\mF^{\prime}) +  a_i \mF \cdot (-a_i \frac{-\kappa_0}{x_{ij}^{a_i+1}} + \mF^{\prime\prime}) }{(\frac{-\kappa_0}{x_{ij}^{a_i}} +\mF^{\prime})^2}  \\
&=& a_i\cdot \frac{ \kappa_0 x_{ij}^{a_i-1}  \left( x_{ij}\mF^{\prime} + a_i  \mF \right) - x_{ij}^{2a_i} (\mF^{\prime})^2 +  x_{ij}^{2a_i}\mF \cdot \mF^{\prime\prime} }{( -\kappa_0   +x_{ij}^{a_i}\mF^{\prime})^2} \\
&<& a_i\cdot  \frac{\kappa_0  x_{ij}^{a_i-1} \left( x_{ij}\mF^{\prime} + a_i  \mF \right) - x_{ij}^{2a_i} (\mF^{\prime})^2  }{( -\kappa_0   +x_{ij}^{a_i}\mF^{\prime})^2} \, ,
\end{eqnarray*}
where the last estimate follows from the fact that, by  Assumption~\ref{ass_prob}, it holds $\mF^{\prime\prime}<0$. 
If $\kappa_0 =0$, then it readily follows that $\frac{\partial}{\partial x_{ij}}\mH_{ij}(x_{ij}+\ji; \kappa_0)<0$ and therefore 
$\mH_{ij}$ is strictly decreasing; thus $\mG_{ij}$ is strictly decreasing as well.  
So we may assume that $\kappa_0>0$. 
If $ x_{ij}\mF^{\prime} + a_i  \mF<0$, then it also follows that $\mH_{ij}$ is strictly decreasing; thus we may also assume that $A:= x_{ij}\mF^{\prime} + a_i  \mF\ge 0$. Now notice that $\frac{\partial A}{\partial x_{ij}} = \mF^{\prime} + x_{ij} \mF^{\prime\prime} + a_i \mF^{\prime} <0$,  
and define the function
\[
H(x_{ij}) := \kappa_0  x_{ij}^{a_i-1} \cdot A - x_{ij}^{2a_i} (\mF^{\prime})^2  \, ;
\]
hence it holds $\frac{\partial}{\partial x_{ij}}\mH_{ij}(x_{ij}+\ji; \kappa_0) < a_i\cdot \frac{H(x_{ij})}{(-\kappa_0 + x_{ij}^{a_i} \mF^{\prime})^2}$.
Moreover, it holds 
\[
\frac{\partial}{\partial x_{ij}}H(x_{ij}) = (a_i-1)\kappa_0 x_{ij}^{a_i-2}\cdot A + \kappa_0x_{ij}^{a_i-1} \cdot\frac{\partial A}{\partial x_{ij}} - 2a_ix_{ij}^{2a_i-1} (\mF^{\prime})^2 - 2x_{ij}^{2a_i} \mF^{\prime} \mF^{\prime\prime} \, .
\]
Since $a_i\le 1$, $A\ge 0$ and $\mF^{\prime},\mF^{\prime\prime},\frac{\partial A}{\partial x_{ij}}<0$, it readily follows that all addends in the previous equation are negative, and therefore   
$\frac{\partial}{\partial x_{ij}}H(x_{ij}) <0$. In other words, $H(\cdot)$ is strictly decreasing in $[0,\omega_{ij}-\ji]$ and, since $H(0)=0$, $H(\omega_{ij}-\ji)<0$, we conclude that $H(x_{ij})\le 0$, for all $x_{ij} \in [0,\omega_{ij}-\ji]$. This implies that $\frac{\partial}{\partial x_{ij}}\mH_{ij}(x_{ij}+\ji; \kappa_0)<0$ for $x_{ij}\in [0,\omega_{ij}-\ji]$.  Since 
$\frac{\partial}{\partial \xtj}\mH_{ij}(\xtj; \kappa_0)=\frac{\partial}{\partial x_{ij}}\mH_{ij}(x_{ij}+\ji; \kappa_0)$, and similarly for $\mG_{ij}$, we conclude that both $\mG_{ij}$ and $\mH_{ij}$ are strictly decreasing in the interval $[0,\omega_{ij}]$, as desired. 
\end{proof}

\section{Proof of Theorem~\ref{one_CPR}}\label{sec:6}

In this section we prove Theorem~\ref{one_CPR}. 
We first introduce some  notation. 
Consider a GNE, say $\xb=(\xb_1,\ldots,\xb_n)\in\mC_n$, where $\xb_i =(x_{i1},\ldots,x_{im})\in C_m$, of a Fragile multi-CPR 
Game satisfying Assumption~\ref{ass_prob}. 
Given $j\in [m]$, let  
\begin{equation}\label{def:support}
\mS(\xtj) = \{i\in [n] : \xtj < \omega_{ij} \text{ and } x_{ij}>0 \}
\end{equation}
be the \emph{support} of the $j$-th CPR and let 
\begin{equation}\label{def:support1}
\mS_I(\xtj)= \{i\in \mS(\xtj) : \xb_i \text{ is of Type~I}\} 
\end{equation}
be the \emph{support of Type~I}, consisting of those players in the support of the $j$-th CPR  whose best response is of Type~I, and 
\begin{equation}\label{def:support2}
\mS_{II}(\xtj)= \{i\in \mS(\xtj) : \xb_i \text{ is of Type~II}\} 
\end{equation}
be the \emph{support of Type~II}, consisting of those players in the support of the $j$-th CPR whose best response is of Type~II. 
Clearly, in view of Theorem~\ref{kkt_thm}, it holds $\mS(\xtj) = \mS_I(\xtj)\cup \mS_{II}(\xtj)$. 

We employ 
the properties of the auxiliary 
functions in the proof of 
Theorem~\ref{one_CPR}, 
a basic ingredient of which is the fact 
that in the setting of Theorem~\ref{one_CPR} the support 
of Type~II is empty. 
The proof is similar to the proof of Theorem~\ref{Hota}, given in~\cite[p.~155]{Hota}. In fact, we prove a bit more. We show that 
Theorem~\ref{one_CPR} is a consequence of  the following result. 

\begin{theorem}\label{gen_one_CPR}
Consider a Fragile multi-CPR Game with $n\ge 1$ 
players and $m\ge 1$ CPRs satisfying Assumption~\ref{ass_prob}. 
Then there exists at most one 
GNE $\xb=(\xb_1,\ldots,\xb_n)$ for which  $\xb_i$ is of Type~I, for all $i\in [n]$.
\end{theorem}
\begin{proof}
Let $\xb=(\xb_1,\ldots,\xb_n)$ be a GNE
such that $\xb_i$ is of Type~I, for all $i\in [n]$ and note that, 
since $\xb_i$ is of Type~I, it holds  
$\sum_j x_{ij} < 1$, for all $i\in [n]$. 
For each $j\in [m]$, let 
$\mS_0(\xtj) :=\{i\in [n]: \xtj<\omega_{ij}\}$. 
We  claim that $\mS_0(\xtj)=\mS(\xtj)$. 
Indeed, if there exists $i\in \mS_0(\xtj)\setminus\mS(\xtj)$ then $x_{ij}=0$ 
and since it holds  
$\xtj<\omega_{ij}$ and $\sum_j x_{ij} < 1$, it 
follows that player $i$ could 
increase her utility by investing a suitably 
small amount, say $\eps>0$, in the $j$-th CPR. But then this implies  that $\xb$ cannot be a GNE, a contradiction. Hence $\mS_0(\xtj)=\mS(\xtj)$.

We  claim that for any two distinct GNEs, 
say $\xb=(\xb_1,\ldots,\xb_n)$ and $\yb=(\yb_1,\ldots,\yb_n)$, for which $\xb_i,\yb_i$ are of Type~I for all $i\in [n]$, it holds $\xtj = \ytj$, for all $j\in [m]$. Indeed, if the claim is not true, then there 
exists $j\in [m]$ such that $\xtj\neq\ytj$. 
Suppose, without loss of generality, that 
$\xtj<\ytj$. 

Since $\xb_i$ is of Type~I, for all $i\in [n]$, 
it follows that $\mS(\xtj)=\mS_I(\xtj)$ and 
$\mS(\ytj)=\mS_I(\ytj)$. Moreover,  since  
$\xtj<\ytj$ it follows that $\mS_I(\ytj) =\mS_0(\ytj)\subset \mS_0(\xtj) =\mS_I(\xtj)$.

Now notice that~\eqref{fixed_g} implies that  
$\mG_{ij}(\xtj)=x_{ij}$, for all $i\in \mS_I(\xtj)$, and $\mG_{ij}(\ytj)=y_{ij}$, for all $i\in \mS_I(\ytj)$. Since $\mS_I(\ytj)\subset \mS_I(\xtj)$ it holds that 
\begin{equation}\label{wrong:0}
\sum_{i \in \mS_I(\ytj)} \mG_{ij}(\xtj)  \le \xtj 
< \ytj  = \sum_{i \in \mS_I(\ytj)} \mG_{ij}(\ytj) \, .
\end{equation}
However, since $\mG_{ij}$ is strictly decreasing, it follows that 
$\mG_{ij}(\xtj) > \mG_{ij}(\ytj)$, for all $i\in \mS_I(\ytj)$, which contradicts~\eqref{wrong:0}. 
We conclude that $\xtj =\ytj$ and 
$\mS_I(\xtj) = \mS_I(\ytj)$. Finally, given a total investment $\xb$ of the players at a GNE, we claim that the optimal investment of every player on any CPR is unique. Indeed, if a player, say $i\in [n]$, has two optimal investments, say $x<z$, on the $j$-th CPR,  then it holds  
$\mG_{ij}(\xb) = x < z = \mG_{ij}(\xb)$, 
a contradiction. The result follows. 
\end{proof}

Theorem~\ref{one_CPR} is a direct consequence of 
Theorem~\ref{gen_one_CPR}, as we now show.

\begin{proof}[Proof of Theorem~\ref{one_CPR}] 
We know from Theorem~\ref{unique_NE} that the game admits a GNE, and it is therefore enough to show that it is unique. 
Since $m=1$, the first condition in Assumption~\ref{ass_prob} implies that no player 
invests an amount of $1$ in the CPR which in turn  implies
that all coordinates of any GNE are  of Type~I. 
The result follows from Theorem~\ref{gen_one_CPR}. 
\end{proof}

Observe that a basic ingredient in the proof of 
Theorem~\ref{gen_one_CPR} is the fact that 
$\mS(\xt) = \mS_I(\xt), \mS(\yt)=\mS_I(\yt)$ and 
$\mS_I(\yt)\subset \mS_I(\xt)$. 
Moreover, observe that the proof of Theorem~\ref{gen_one_CPR} proceeds in two steps: in the first step it is shown that 
any two GNEs admit the same total investment in the CPR, 
and in the second step it is shown that, given an optimal total investment, every player has a unique optimal investment in the CPR. 
In the following section we are going to improve 
upon the aforementioned observations. 
A bit more concretely, we are going to prove 
that the set consisting of all GNEs of a Fragile multi-CPR Game is ``small" via showing that the set 
consisting of all total investments at the GNEs is ``small".

\section{Proof of Theorem~\ref{nr_of_GNE}}\label{sec:7}

Throughout this section, we denote by $G^{(2)}$ a 
Fragile multi-CPR Game satisfying Assumption~\ref{ass_prob}. 
Moreover, given a finite set, $F$, we denote 
by $|F|$ its cardinality. 
Now consider the set 
\[
\mN(G^{(2)}) := \{\xb \in \mC_n : \xb \text{ is a GNE of } G^{(2)}\} 
\]
and, given $\xb=(\xb_1,\ldots,\xb_n)\in \mN(G^{(2)})$, let 
\[
\mT_I(\xb) = \{i\in [n] : \xb_i \text{ is of Type~I}\}
\]
and 
\[
\mT_{II}(\xb) = \{i\in [n] : \xb_i \text{ is of Type~II}\} \, .
\]
Recall the definition of 
active CPRs corresponding 
to $\xb_{-i}$, which is denoted $A(\xb_{-i})$
and is defined in~\eqref{active_cpr}, as well as 
the definition of 
effective 
CPRs corresponding to 
$\xb_{-i}$, which is denoted  $J_{\xb_{-i}}$ and is 
defined in~\eqref{J_set}.

\begin{lemma}\label{lem:antichain1}
Let $\xb=(\xb_1,\ldots,\xb_n)\in\mN(G^{(2)})$ and suppose that $i\in\mT_I(\xb)$, for some $i\in [n]$. 
Then it holds $J_{\xb_{-i}}=A(\xb_{-i})$. 
\end{lemma}
\begin{proof}
Recall from Theorem~\ref{kkt_thm} that $J_{\xb_{-i}}$ is such that $x_{ij}>0$ if and only if $j\in J_{\xb_{-i}}$.
Suppose, towards arriving at a contradiction,
that there exists $j\in A(\xb_{-i})\setminus J_{\xb_{-i}}$. Since $i\in\mT_I(\xb)$, it follows that  $\sum_{j\in J_{\xb_{-i}}} x_{ij}<1$ and thus player $i$ can increase her utility 
by investing a suitably small amount $\eps>0$ in the $j$-th CPR. This  contradicts the fact that $\xb$ is a GNE, and the lemma follows. 
\end{proof}

\begin{lemma}\label{lem:antichain2}
Let $\xb=(\xb_1,\ldots,\xb_n)$ and $\yb=(\yb_1,\ldots,\yb_n)$ be two elements from 
$\mN(G^{(2)})$ such that $\xtj\le \ytj$, for all $j\in [m]$. Then the following hold true:
\begin{enumerate}  
\item If $i\in \mT_I(\xb)$, then  $J_{\yb_{-i}}\subset J_{\xb_{-i}}$.   
\item It holds  $\mT_I(\xb) \subset \mT_{I}(\yb)$. 
\end{enumerate}
\end{lemma}
\begin{proof} 
Fix $i\in [n]$ such that $i\in\mT_I(\xb)$ and  
notice that Lemma~\ref{lem:antichain1} implies that $\xtj\ge \omega_{ij}$, for all $j\in [m]\setminus J_{\xb_{-i}}$. 
Since $\xtj\le \ytj$, for all $j\in [m]$,  it  
holds $\ytj\ge \omega_{ij}$, for all $j\in [m]\setminus J_{\yb_{-i}}$, and 
we conclude that $J_{\yb_{-i}}\subset J_{\xb_{-i}}$. The first statement follows. 

We proceed with the second statement.
Let $i\in [n]$ be such that $\xb_i$ is of Type~I. 
We have to show that $\yb_i$ is also of Type~I.
Suppose that this is not true; hence $\yb_i$ is of Type~II, and thus it holds 
$\sum_{j\in J_{\yb_{-i}}}y_{ij} =1$. 
Since $\yb_i$ is of Type~II, it follows from~\eqref{fixed_h} that $\mH_{ij}(\ytj; \kappa_0)=y_{ij}$, for all $j\in J_{\yb_{-i}}$ and some $\kappa_0\ge 0$. 
Since $\xb_i$ is of Type~I and $\mH_{ij}$ is decreasing, we may apply~\eqref{fixed_g} and conclude 
\[
x_{ij} = \mG_{ij}(\xtj) \ge \mH_{ij}(\xtj;\kappa_0) \ge \mH_{ij}(\ytj;\kappa_0)  = y_{ij},\, \text{ for all }\, j\in J_{\yb_{-i}} \, .
\]
Hence $1 > \sum_{j\in J_{\yb_{-i}}} x_{ij}\ge \sum_{j\in J_{\yb_{-i}}} y_{ij} = 1$, a contradiction. 
The result follows. 
\end{proof}

\begin{lemma}\label{lem:antichain3}
Assume that $m\le n$. Then it holds $\mT_I(\xb)\neq \emptyset$, for every $\xb\in \mN(G^{(2)})$.
\end{lemma}
\begin{proof}
Suppose that the conclusion is not true; hence there 
exists $\xb=(\xb_1,\ldots,\xb_n)\in \mN(G^{(2)})$ such that 
$\mT_{II}(\xb)= [n]$, which in turn implies that
$\sum_{j\in [m]} \xtj = n \ge m$. Hence there exists 
$k\in [m]$ such that $\xb_T^{(k)}\ge 1$. 
We now claim that $\xb_T^{k|i}\ge \omega_{ik}$, for all $i\in\mS_{II}(\xb_T^{(k)})$, where $\mS_{II}(\cdot)$ is defined in~\eqref{def:support2} and $\omega_{ik}$ is given by Lemma~\ref{VTP_lemma}. To prove the claim, 
notice that if there exists $i\in\mS_{II}(\xb_T^{(k)})$ 
such that $\xb_T^{k|i}< \omega_{ik}$ then, since 
$x_{ik}$ is a best response of player $i$ in the $k$-th CPR, by Remark~\ref{br_NE}, it would assume a value 
for which $x_{ik}+ \xb_T^{k|i}<\omega_{ij}$, which contradicts the fact that $\xb_T^{(k)}\ge 1$. 
The claim follows. 

However, since $x_{ik}$ is a best response and $\xb_T^{k|i}\ge \omega_{ik}$, for all $i\in\mS_{II}(\xb_T^{(k)})$, it follows that 
$x_{ik}=0$, for all $i\in\mS_{II}(\xb_T^{(k)})$. 
This contradicts the fact that $\xb_T^{(k)}\ge 1$, and the result follows. 
\end{proof}

\begin{lemma}\label{lem:antichain5}
Assume that $m\le n$. Then 
there do not exist distinct elements  $\xb=(\xb_1,\ldots,\xb_n)$ and $\yb=(\yb_1,\ldots,\yb_n)$ in 
$\mN(G^{(2)})$ for which it holds  
$\xtj\le \ytj$, for all $j\in [m]$, and $\sum_{j\in [m]}\xtj < \sum_{j\in [m]}\ytj$. 
\end{lemma}
\begin{proof}
Suppose that such GNEs do exist. 
Since $m\le n$, it follows from Lemma~\ref{lem:antichain3} that $\mT_I(\xb)\neq \emptyset$, for every $\xb\in \mN(G^{(2)})$. 

Notice that Lemma~\ref{lem:antichain2} implies that  $\mT_I(\xb)\subset\mT_I(\yb)$ and $J_{\yb_{-i}}\subset J_{\xb_{-i}}$, and~\eqref{fixed_g} 
implies that $\mG_{ij}(\xtj)=x_{ij}$, for all $i\in\mT_I(\xb)$ and all $j\in J_{\xb_{-i}}$. 
Similarly, it holds 
$\mG_{ij}(\ytj)=y_{ij}$, for all $i\in\mT_I(\yb)$ and all $j\in J_{\yb_{-i}}$. 
Hence we may write  
\[
\sum_{j\in [m]} \xtj = \sum_{i\in\mT_I(\xb)} \sum_{j\in J_{\xb_{-i}}} \mG_{ij}(\xtj) + |\mT_{II}(\xb)| 
\]
as well as  
\[
\sum_{j\in [m]} \ytj = \sum_{i\in\mT_I(\yb)} \sum_{j\in J_{\yb_{-i}}} \mG_{ij}(\ytj) + |\mT_{II}(\yb)|  \, .
\]
Since $\sum_j\xtj < \sum_j\ytj$, $|\mT_I(\xb)|\le |\mT_I(\yb)|$ and  $|\mT_{II}(\xb)|\ge |\mT_{II}(\yb)|$ hold true, it follows that 
\begin{eqnarray}\label{eq:xxyy}
\sum_{i\in\mT_I(\xb)}\, \sum_{j\in J_{\xb_{-i}}} \mG_{ij}(\xtj) + |\mT_{II}(\xb)\setminus \mT_{II}(\yb)| &<& \sum_{i\in\mT_I(\xb)}\, \sum_{j\in J_{\yb_{-i}}} \mG_{ij}(\ytj) \nonumber \\
&+& \sum_{i\in\mT_I(\yb)\setminus \mT_I(\xb)}\,\, \sum_{j\in J_{\yb_{-i}}} \mG_{ij}(\ytj) \, .
\end{eqnarray}
However, the fact that $\mG_{ij}$ is decreasing 
implies that $\mG_{ij}(\xtj)\ge \mG_{ij}(\ytj)$, for 
all $i\in \mT_I(\xb)$ and all $j\in J_{\yb_{-i}}$; 
hence it holds 
\begin{equation}\label{eq:xx}
\sum_{i\in\mT_I(\xb)}\, \sum_{j\in J_{\xb_{-i}}} \mG_{ij}(\xtj) \ge \sum_{i\in\mT_I(\xb)}\, \sum_{j\in J_{\yb_{-i}}} \mG_{ij}(\ytj) \, .
\end{equation}
Moreover, since $\sum_{j\in J_{\yb_{-i}}} \mG_{ij}(\ytj)<1$, for all $i\in\mT_I(\yb)\setminus \mT_I(\xb)$, it holds 
\begin{equation}\label{eq:xy}
\sum_{i\in\mT_I(\yb)\setminus \mT_I(\xb)}\,\, \sum_{j\in J_{\yb_{-i}}} \mG_{ij}(\ytj) < |\mT_I(\yb)\setminus \mT_I(\xb)| = |\mT_{II}(\xb)\setminus \mT_{II}(\yb)| \, .
\end{equation}
Now notice that~\eqref{eq:xx} and~\eqref{eq:xy} contradict~\eqref{eq:xxyy}. The result follows. 
\end{proof}

Finally, the proof of Theorem~\ref{nr_of_GNE} 
requires the following measure-theoretic results. 
Here and later, given 
a positive integer $k\ge 1$, $\mL^k$ denotes $k$-dimensional Lebesgue measure. Moreover, 
given a function $f:\mathbb{R}^k\to\mathbb{R}^m$ 
and a set $B\subset \mathbb{R}^m$, we denote 
$f^{-1}(B) := \{\xb\in\mathbb{R}^k: f(\xb)\in B\}$ 
the preimage of $B$ under $f$. 

\begin{lemma}\label{0_property}
Let $f:\mathbb{R}^d\to\mathbb{R}^m$ be a continuously differentiable function for which 
$\mL^d(\{\xb\in\mathbb{R}^d : \nabla f(\xb)=0\})=0$. 
Then it holds $\mL^d(f^{-1}(A))=0$, for every $A\subset\mathbb{R}^m$  for which $\mL^m(A)=0$. 
\end{lemma}
\begin{proof}
See~\cite[Theorem~1]{Ponomarev}. 
\end{proof}

Let $m\ge 1$ be an integer. 
A set $A\subset [0,1]^m$ is called an 
\emph{antichain} if it does not contain two distinct elements 
$\xb=(x_1,\ldots,x_m)$ and $\yb=(y_1,\ldots,y_m)$ 
such that $x_j\le y_j$, for all $j\in [m]$.

\begin{lemma}\label{sperner}
Let $A\subset [0,1]^m$ be an antichain. 
Then $\mL^m(A)=0$.
\end{lemma}
\begin{proof}
The result is an immediate consequence of  Lebesgue's density theorem. Alternatively, it follows from the main result in~\cite{Engel}, and  from~\cite[Theorem~1.3]{EMPR}.
\end{proof}
 
Now given $\xb=(\xb_1,\ldots,\xb_n)\in\mC_n$, let 
$\mathbf{v}_{\xb}$ denote the vector 
\begin{equation}\label{v_x}
\mathbf{v}_{\mathbf{x}} := (\xb_T^{(1)},\ldots,\xb_T^{(m)})\in [0,1]^m,  
\end{equation}
where 
$\xtj, j\in [m]$, is defined in~\eqref{sum_xs}. 
Finally, given $N\subset \mC_n$, define the set 
\begin{equation}\label{Wn_set}
W_N := \bigcup_{\xb\in N} \mathbf{v}_{\xb} \, .
\end{equation}

The proof of Theorem~\ref{nr_of_GNE} is almost complete. 

\begin{proof}[Proof of Theorem~\ref{nr_of_GNE}]
To simplify notation, let us set $N:=\mN(G^{(2)})$.  We have to show that 
$\mL^{nm}(N)=0$. 

Now let $f$ denote the map $f:\mC_n \to [0,1]^m$  
given by $f(\xb) = \mathbf{v}_{\xb}$, where 
$\mathbf{v}_{\xb}$ is defined in~\eqref{v_x}. 
It is straightforward to verify that 
$\{\xb\in\mC_n :\nabla f(\xb)=0 \}=\emptyset$. 

Now consider the set $W_N$, defined in~\eqref{Wn_set}, 
and notice that 
Lemma~\ref{lem:antichain5} implies that 
$W_N$ is an antichain; hence it follows 
from Lemma~\ref{sperner} that 
$\mL^{m}(W_N)=0$. Therefore, Lemma~\ref{0_property} yields
\[
\mL^{nm}(N) = \mL^{nm}(f^{-1}(W_N)) = 0 \, , 
\]
as desired. 
\end{proof}

\section{A restricted version of the  game}\label{sec:7.5}

Let $G^{(2)}$ denote a 
Fragile multi-CPR Game satisfying Assumption~\ref{ass_prob}. 
In this section we show that $G^{(2)}$ admits finitely many 
GNEs, subject to the constraint that 
the total investment in each CPR is fixed. We then use this result, in the next section, in order to formulate 
a conjecture which is equivalent to Conjecture~\ref{conj:1}. 
Before being more precise, we need some extra 
piece of notation.

Given a  set $F\subset [m]$ and 
real numbers $\{r_j\}_{j\in F}\subset [0,1]$, 
indexed by $F$, 
we denote by $W(\{r_j\}_{j\in F})$ the set  
\[
W(\{r_j\}_{j\in F}) := \{\xb =(\xb_1,\ldots,\xb_n)\in\mC_n : \xtj = r_j, \text{ for } j \in F \} \, ,
\]
where $\xtj$ is defined in~\eqref{sum_xs}. 
In other words, $W(\{r_j\}_{j\in F})$ consists 
of those strategy profiles for which the total 
investment in the 
CPRs corresponding to elements in $F$ is fixed, 
and equal to the given numbers $\{r_j\}_{j\in F}$.

In this section we prove the following. 

\begin{theorem}\label{w_set}
Fix real numbers $r_1,\ldots,r_m\in [0,1]$. Then the set 
$W:=W(r_1,\ldots,r_m)$ contains at most $2^{n\cdot (m+1)}$ GNEs of $G^{(2)}$.  
\end{theorem}

The proof requires a couple of observations which we collect 
in the following lemmata. 

\begin{lemma}\label{w00_set}
Suppose that $\xb=(\xb_1,\ldots,\xb_n)$ and $\yb=(\yb_1,\ldots,\yb_n)$ are two GNEs of $G^{(2)}$ such that $\xb,\yb \in W:=W(r_j)$ and   
$0<x_{ij}<y_{ij}$, for
some $i\in [n]$, $j\in [m]$ and $r_j\in [0,1]$. 
Then either $\xb_i$ is of Type~II or $\yb_i$ is of Type~II. 
\end{lemma}
\begin{proof}
Suppose, towards arriving at a contradiction, that the conclusion is not true. 
Then both $\xb_i$ and $\yb_i$ are of Type I, and 
thus~\eqref{fixed_g} implies that 
$\mG_{ij}(r)=x_{ij}$ and $\mG_{ij}(r)=y_{ij}$.  Hence it holds $\mG_{ij}(r_j)=x_{ij}< y_{ij}=\mG_{ij}(r_j)$, a contradiction. 
The result follows. 
\end{proof}

\begin{lemma}\label{w0_set}
Suppose that $\xb=(\xb_1,\ldots,\xb_n)$ and $\yb=(\yb_1,\ldots,\yb_n)$ are two GNEs of $G^{(2)}$ for which it holds $\xb,\yb \in W:=W(r_l,r_{\ell})$  
and $0<x_{ij}<y_{ij}$ and $x_{i\ell}>y_{i\ell}>0$, for some $i\in [n]$ and $\{j, \ell\}\subset  [m]$. 
Then either 
$\xb_i$ is of Type~I or $\yb_i$ is of Type~I.
\end{lemma}
\begin{proof}
Suppose, towards arriving at a contradiction, that both $\xb_{-i}$ and $\yb_{-i}$ are of Type~II. 
Recall the definition of $\psi_{ij}(\,\cdot\, ;\, \cdot\,)$, given in~\eqref{psi_function}, and notice that, since both $\xb_i,\yb_i$ are of Type~II,  Theorem~\ref{kkt_thm} implies the existence of $\kappa_x,\kappa_y\ge 0$ such that 
\[
\kappa_x =x_{ij}^{a_i-1}\cdot\psi_{ij}(x_{ij}; r_j-x_{ij}) = x_{i\ell}^{a_i-1}\cdot\psi_{i\ell}(x_{i\ell}; r_{\ell}-x_{i\ell})
\]
and 
\[
\kappa_y = y_{ij}^{a_i-1}\cdot\psi_{ij}(y_{ij}; r_j-y_{ij}) = y_{i\ell}^{a_i-1}\cdot \psi_{i\ell}(y_{i\ell}; r_{\ell}-y_{i\ell}) \, .
\]
Now notice that, for all $k\in [m]$, the function $\Psi_k(x):=x^{a-1}\cdot\psi_{ik}(x; r-x)$ is decreasing in $x$, for fixed $r>0$ and $a\in (0,1]$. 
Hence,  $x_{ij}<y_{ij}$ implies that $\kappa_x > \kappa_y$, and $x_{i\ell}>y_{i\ell}$ implies that 
$\kappa_x < \kappa_y$, a contradiction. 
The result follows. 
\end{proof}

We may now proceed with the proof of the main result of this section. 

\begin{proof}[Proof of Theorem~\ref{w_set}]
For every $i\in [n]$, define the set 
\[
N_i := \{\xb_i\in C_m : (\xb_i, \xb_{-i})\in W, \text{ for some } \xb_{-i}\in\mC_{-i}\} \, .
\]
We first show that the cardinality of $N_i$, denoted $|N_i|$, is at most $2^{m+1}$. 

Let $\xb_i\in N_i$, and 
recall from Theorem~\ref{kkt_thm}, and~\eqref{J_set}, that 
there exists $J\subset [m]$ such that 
$x_{ij}>0$, when $j\in J$, and $x_{ij}=0$ when $j\in [m]\setminus J$. 
In other words, to every $\xb_i\in N_i$ there corresponds a set $J\subset [m]$ such that 
$x_{ij}>0$ if and only if $j\in J$. 
Now, given $J\subset [m]$, let 
\[
N_J := \{\xb_i\in N_i : x_{ij}>0 \text{ if and only if } j\in J\} \, .
\]
Assume first that $|J|\ge 2$. In this case we 
claim that $|N_J|\le 2$.  Indeed, if $|N_J|\ge 3$, then there are two elements,  
say $\xb^{(1)},\xb^{(2)}\in N_J$, which are either both of Type~I, or both of Type~II. 
If both $\xb^{(1)}$ and $\xb^{(2)}$ are of Type~I, then there exists $j\in J$ such that, without loss of generality, it holds $x_{ij}^{(1)} < x_{ij}^{(2)}$; which contradicts Lemma~\ref{w00_set}. 
If both $\xb^{(1)}$ and $\xb^{(2)}$ are of Type~II, then there exist $j,\ell\in J$ such that 
$x_{ij}<y_{ij}$ and $x_{i\ell}> y_{i\ell}$; 
which contradicts Lemma~\ref{w0_set}.
The claim follows. 

If $|J|=1$, say $J=\{j\}$, we claim that $|N_J|\le 1$. 
Indeed, suppose that $|N_J|\ge 2$ holds true and notice that every element of $N_J$ is of Type~I. 
However, the assumption that $|N_J|\ge 2$ implies that there exist $\xb^{(1)},\xb^{(2)}\in N_J$ such that $0<x_{ij}^{(1)}<y_{ij}^{(1)}$; which contradicts Lemma~\ref{w00_set}. The second claim follows.  

Since there are $2^m$ subsets $J\subset [m]$, and for each $J$ it holds $|N_J|\le 2$,  it follows that there are at most $2^{m+1}$ elements in $N_i$. 
Since there are $n$ players  in the game, the result follows. 
\end{proof}

\section{Concluding remarks and conjectures}\label{sec:8}

Let $G^{(2)}$ denote a Fragile multi-CPR Game satisfying 
Assumption~\ref{ass_prob}, and let 
$\mN(G^{(2)})$ be the set consisting of all GNEs 
of $G^{(2)}$. 
So far we have proven that the $(n\cdot m)$-dimensional 
Lebesgue measure of $\mN(G^{(2)})$ equals zero, 
but there are several problems and questions 
that remain open. 
First and foremost, we believe that the following  holds true.

\begin{conjecture}\label{conj:2}
Let $N:=\mN(G^{(2)})$. Then the 
antichain $W_N$, defined in~\eqref{Wn_set}, is finite. 
\end{conjecture}

Notice that if Conjecture~\ref{conj:2} holds true then, 
in view of Theorem~\ref{w_set}, 
Conjecture~\ref{conj:1} holds true as well. 
Since the converse is clearly true, it follows that 
Conjecture~\ref{conj:1} and Conjecture~\ref{conj:2} 
are equivalent. 
The exact number of GNEs in a Fragile multi-CPR Game appears to depend on
the relation between the number of players, $n$, 
and the number of CPRs, $m$. 
When $n\ge m$ we conjecture that that for every GNE
the players choose best responses of Type~I and therefore, provided this is indeed the case,  
Theorem~\ref{gen_one_CPR} would imply that 
the game admits a unique GNE. 

\begin{conjecture}\label{conj:3}
If $n\ge m$, then $|\mN(G^{(2)})| =1$.  
\end{conjecture}

Another line of research is to investigate the \emph{best response 
dynamics} of a Fragile multi-CPR Game,  
which may be seen as a behavioral rule along  
which  players fix an initial investment in the  CPRs and proceed with updating their investment,  over rounds, in such a way that in the $t$-th round player $i\in [n]$ 
invests $\bb_i^{(t)} :=B_i(\xb_{-i}^{(t)})$, where 
$B_i(\cdot)$ is defined in~\eqref{def_BR} and 
$\xb_{-i}^{(t)}\in\mC_{-i}$ is the strategy profile of all  players except player $i$ in the $t$-th round. A natural question to ask is whether the best response dynamics converge, i.e., whether there exists a round $t_0$ such that $\bb_i^{(t)} = \bb_i^{(t_0)}$, for all $t\ge t_0$ and all $i\in [n]$. 

\begin{conjecture}\label{conj:4}
The best response dynamics of $G^{(2)}$ converge. 
\end{conjecture}

When $m=1$, it is shown in~\cite{Hota} that the best 
response dynamics of the Fragile CPR Game converge to its  Nash Equilibrium. This is obtained as a 
consequence of the fact that the best response correspondence is single-valued and 
decreasing  in the total investment in the CPR  
(see the remarks following~\cite[Proposition~7]{Hota}). 
Moreover, it is not difficult to verify that the 
Nash equilibrium of the Fragile CPR Game is also 
the Generalized Nash equilibrium. Hence, 
the best response dynamics of a Fragile CPR Game converge 
to the Generalized Nash equilibrium. 
When $m\ge 2$, the best response correspondence need 
no longer be decreasing in each CPR.  
It is decreasing for those players whose best response is 
of Type~I, as can be easily seen using the fact that 
the auxiliary function $\mG_{ij}$ is decreasing. 
This monotonicity may no longer be true when a  
player moves from a best response of Type~II to a best response 
of Type~I, or from a best response of Type~II 
to a best response of the same type. Furthermore, Theorem~\ref{kkt_thm} does not guarantee that the set of effective CPR, defined in~\eqref{J_set}, is unique. Hence, the best response correspondence may not be single-valued. 
So far our theoretical analysis does not provide sufficient evidence 
for the holistic validity of Conjecture~\ref{conj:4}. However,  
our numerical experiments suggest that 
Conjecture~\ref{conj:4} holds true, and we expect that we will be able 
to report on that matter in the future.

\end{document}